\documentclass[10pt,journal,compsoc]{IEEEtran}
\usepackage{ifpdf}
\usepackage{mathrsfs}
\ifCLASSOPTIONcompsoc
\usepackage[nocompress]{cite}
\else
\usepackage{cite}
\fi
\usepackage{amsmath}
\usepackage{amssymb}
\usepackage{bm}
\usepackage{array}
\usepackage{comment}
\usepackage{algpseudocode}
\usepackage{graphicx}
\usepackage{xcolor}
\usepackage{color}
\usepackage{booktabs,hyphenat,balance}
\usepackage{threeparttable}
\usepackage{microtype}
\usepackage{multirow}
\usepackage{arydshln}
\usepackage{url}
\usepackage{multicol}
\usepackage{wrapfig}
\usepackage{amsthm}
\usepackage{amsmath}

\newtheorem{lemma}{Lemma}
\usepackage[bookmarks,backref=true,linkcolor=black]{hyperref}
\hypersetup{
	pdfauthor = {},
	pdftitle = {},
	pdfsubject = {},
	pdfkeywords = {},
	colorlinks=true,
	linkcolor= black,
	citecolor= black,
	pageanchor=true,
	urlcolor = black,
	plainpages = false,
	linktocpage
}
\usepackage{amsthm}
\usepackage{wrapfig}
\usepackage{enumitem}
\graphicspath{{./images/}}

\newtheorem{proposition}{Proposition}
\newtheorem{assumption}{Assumption}

\def \etal {{\emph{et al}.\thinspace}}

\newcommand{\myparagraph}[1]{ \vspace{0.2\baselineskip} \noindent \textbf{#1.}~}
\newcommand{\firstparagraph}[1]{\noindent \textbf{#1.}~}
\newcommand{\loss}[1]{E_{\text{#1}}}
\newcommand{\weight}[1]{w_{\text{#1}}}
\newcommand{\lvisdat}{\loss{v,dat}}
\newcommand{\lvissmooth}{\loss{v,smo}}
\newcommand{\lviscons}{\loss{v,comp}}
\newcommand{\lvisdepth}{\loss{v,dep}}
\newcommand{\wvissmooth}{\weight{v,smo}}
\newcommand{\wviscons}{\weight{v,comp}}
\newcommand{\wvisdepth}{\weight{v,dep}}
\newcommand{\loccdat}{\loss{o,dat}}
\newcommand{\loccsmooth}{\loss{o,smo}}
\newcommand{\locccons}{\loss{o,comp}}
\newcommand{\loccdepth}{\loss{o,dep}}
\newcommand{\woccsmooth}{\weight{o,smo}}
\newcommand{\wocccons}{\weight{o,comp}}
\newcommand{\woccdepth}{\weight{o,dep}}
\newcommand{\lgc}{\loss{cons}}
\newcommand{\wgc}{\weight{cons}}
\newcommand{\lbsr}{\loss{surf}}
\newcommand{\lbsrdat}{\loss{surf,dat}}
\newcommand{\lbsrfair}{\loss{surf,fair}}
\newcommand{\wbsrfair}{\weight{surf,fair}}
\newcommand{\wsurffair}{\weight{f}}
\newcommand{\cpoint}[1]{\mathbf{p}_{#1}}
\newcommand{\cpxycoord}[1]{\overline{\mathbf{p}}_{#1}}
\newcommand{\candxycoord}[1]{\widetilde{\mathbf{p}}_{#1}}
\newcommand{\xysample}[1]{\mathbf{a}_{#1}}
\newcommand{\boundarysamples}{\mathcal{B}}
\newcommand{\interiorsamples}{\mathcal{I}}

\newcommand{\wcpxysmooth}{\lambda_3}
\newcommand{\lcpxysmooth}{\loss{fair}}
\newcommand{\wsampletocp}{\lambda_1}
\newcommand{\wcptosample}{\lambda_2}
\newcommand{\closestcpidx}[1]{\rho_{#1}}
\newcommand{\closestsampleidx}[1]{\eta_{#1}}
\newcommand{\lcdf}{\loss{cdf}}
\newcommand{\lcdfdat}{\loss{cdf,dat}}
\newcommand{\lcdfsmooth}{\loss{cdf,smo}}
\newcommand{\lcdfcons}{\loss{cdf,cons}}
\newcommand{\wcdfsmooth}{\weight{cdf,smo}}
\newcommand{\wcdfcons}{\weight{cdf,cons}}

\newcommand{\vddatweightdepth}{\sigma_{\text{d}}}
\newcommand{\vddatweightnormal}{\sigma_{\text{n}}}
\newcommand{\vdsmoweightdepth}{\tau_{\text{d}}}
\newcommand{\vdsmoweightnormal}{\tau_{\text{n}}}
\newcommand{\gcweightdepth}{\gamma_{\text{d}}}
\newcommand{\gcweightnormal}{\gamma_{\text{n}}}

\newcommand{\fourneighborset}[1]{N_4(#1)}

\newcommand{\gt}[1]{#1^{\text{gt}}}
\newcommand{\vispixset}{\mathcal{P}_{\text{v}}} %visible surface pixel set
\newcommand{\visnbset}{\mathcal{N}_{\text{v}}} %visible neighboring pixel set
\newcommand{\visxnbset}{\mathcal{R}_{\text{v}}} % x-neighbors set in the visible region
\newcommand{\visynbset}{\mathcal{U}_{\text{v}}} % y-neighbors set in the visible region
\newcommand{\occpixset}{\mathcal{P}_{\text{o}}} %occluded surface pixel set
\newcommand{\occnbset}{\mathcal{N}_{\text{o}}} %occluded neighboring pixel set
\newcommand{\rightneighbor}[1]{r({#1})}
\newcommand{\upneighbor}[1]{u({#1})}
\newcommand{\pixwidth}{\delta}
\newcommand{\vissamplepixset}{\mathcal{S}_{\text{v}}}
\newcommand{\spd}[1]{#1^{\text{s}}} % sample depth
\newcommand{\bsrparamset}{\mathcal{S}_{\text{bsr}}}
\newcommand{\bsrsurf}[1]{\mathbf{s}(#1)}
\newcommand{\bsrsurfgt}[1]{\mathbf{s}_{\text{gt}}(#1)}
\newcommand{\ulaplacian}[1]{\mathbf{l}_{#1}^u}
\newcommand{\vlaplacian}[1]{\mathbf{l}_{#1}^v}
\newcommand{\dpoint}[1]{\mathbf{q}_{#1}}
\newcommand{\dpointset}{\mathcal{D}}
\newcommand{\compx}[1]{{#1}'}
\newcommand{\compxenc}[1]{\hat{#1}}
\newcommand{\heightfield}[1]{f(#1)}
\newcommand{\zfair}[1]{E_f(#1)}
\newcommand{\chamferdist}{D_{\text{c}}}
\newcommand{\normaldist}{D_{\text{n}}}

\newcommand{\recsurfsample}[1]{\mathbf{x}_{#1}}
\newcommand{\gtsurfsample}[1]{\mathbf{y}_{#1}}
\newcommand{\gtsurfcpidx}[1]{\alpha_{#1}}
\newcommand{\recsurfcpidx}[1]{\beta_{#1}}
\newcommand{\recsurfsampleset}{\mathcal{S}}
\newcommand{\gtsurfsampleset}{\mathcal{S}_{\text{gt}}}
\newcommand{\flsampleset}{\mathcal{S}_{\text{f}}}
\newcommand{\flsample}[1]{\mathbf{f}_{#1}}
\newcommand{\pqedgedist}{D_{\text{f}}}
\newcommand{\anglefunc}{d_{\text{n}}}
\newcommand{\vecsymbol}{\mathbf{v}}

\newcommand{\meanplanarity}{P_{\text{mean}}}
\newcommand{\maxplanarity}{P_{\text{max}}}
\newcommand{\avgmeanplanarity}{\overline{P}_{\text{mean}}}
\newcommand{\avgmaxplanarity}{\overline{P}_{\text{max}}}
\newcommand{\avgpqedgedist}{\overline{D}_{\text{f}}}

\newcommand{\edgecp}[1]{\mathbf{g}_{\zeta_{#1}}}
\newcommand{\cdfsmoenergy}{E_{\text{smooth}}}

\newcommand{\visdepthmse}{D_{\text{v}}}
\newcommand{\occdepthmse}{D_{\text{o}}}

\newcommand{\copyrightnote}{
	{
		\begin{picture}(0,0)(0,0)
			\put(-475,-508){ 
				\parbox{\textwidth}{%
					\footnotesize{
						\fontfamily{cmr}\selectfont
						\textcopyright~2022 IEEE.  Personal use of this material is permitted.  Permission from IEEE must be obtained for all other uses, in any current or future media, including reprinting/republishing this material for advertising or promotional purposes, creating new collective works, for resale or redistribution to servers or lists, or reuse of any copyrighted component of this work in other works.
					}
				}
			}
		\end{picture}
	}
}

\def \mn {\mathbf{n}}
\def \mt {\mathbf{t}}

\begin{document}

\title{Sketch2PQ: Freeform Planar Quadrilateral Mesh Design via a Single Sketch}
\author{
 Zhi~Deng, Yang~Liu, Hao~Pan, Wassim~Jabi, Juyong~Zhang, Bailin~Deng$^\dagger$
 \IEEEcompsocitemizethanks{
 \IEEEcompsocthanksitem Z.~Deng is with School of Data Science, University of Science and Technology of China.
 \IEEEcompsocthanksitem H.~Pan and Y.~Liu are with Internet Graphics Group, Microsoft Research Asia.
 \IEEEcompsocthanksitem W.~Jabi is with the Welsh School of Architecture, Cardiff University.
 \IEEEcompsocthanksitem J.~Zhang is with School of Mathematical Sciences, University of Science and Technology of China.
 \IEEEcompsocthanksitem B.~Deng is with School of Computer Science and Informatics, Cardiff University.
 }% <-this % stops a space
 \thanks{$^\dagger$Corresponding author. Email: \texttt{DengB3@cardiff.ac.uk}.}
 }

\IEEEtitleabstractindextext{%
 \begin{abstract}
The freeform architectural modeling process often involves two important stages: concept design and digital modeling. In the first stage, architects usually sketch the overall 3D shape and the panel layout on a physical or digital paper briefly. In the second stage, a digital 3D model is created using the sketch as a reference. The digital model needs to incorporate geometric requirements for its components, such as the planarity of panels due to consideration of construction costs, which can make the modeling process more challenging. 
In this work, we present a novel sketch-based system to bridge the concept design and digital modeling of freeform roof-like shapes represented as planar quadrilateral (PQ) meshes. Our system allows the user to sketch the surface boundary and contour lines under axonometric projection and supports the sketching of occluded regions. In addition, the user can sketch feature lines to provide directional guidance to the PQ mesh layout.
Given the 2D sketch input, we propose a deep neural network to infer in real-time the underlying surface shape along with a dense conjugate direction field, both of which are used to extract the final PQ mesh. To train and validate our network, we generate a large synthetic dataset that mimics architect sketching of freeform quadrilateral patches. 
The effectiveness and usability of our system are demonstrated with quantitative and qualitative evaluation as well as user studies.
\copyrightnote{}
 \end{abstract}

 }

% make the title area
\maketitle
\section{Introduction} \label{sec:intro}

\IEEEPARstart{F}{reeform}
3D architectural modeling usually starts with the concept design, which can be conveyed via 2D sketching. The sketching depicts the boundary and the profile of the surface, and directional strokes may be added to describe the panel layout (see Fig.~\ref{fig:design}  for two freeform roof designs via 2D sketching).
Afterwards, the 2D sketch needs to be converted into a 3D digital model. 
To do so, architects usually need to use 3D modeling software to manually create surface geometry according to the sketch and spend a considerable amount of time editing it into the desired shape. The created geometry is then decomposed into panels, which often involves computational routines that enforce geometric constraints related to construction requirements and design intent~\cite{pottmann2015architectural}.
This overall workflow can be time-consuming and requires professional knowledge at each step.

An alternative approach is sketch-based modeling~\cite{Olsen:2009:survey} which automatically converts 2D sketches into 3D shapes. The main obstacle in this process is the ambiguity in the mapping from 2D to 3D which makes the problem ill-posed. In addition, the sketches can be inaccurate and noisy, which makes the conversion problem even more challenging.
Over the past two decades, many approaches have been proposed to address these challenges. 
Geometry priors from shape classes~\cite{Nealen:2007:FiberMesh,Chen2008,Gingold2009Structured,Wu:2018:shaped} provide a strong cue to resolve the ambiguity, but it is not easy to generalize them to unseen shape classes such as freeform shapes.  Sketching from multiple views~\cite{Lun:2017:SketchModeling,Delanoy20183D} is a natural way to remove ambiguity and create a freeform design, but it could be inconvenient for the user to change the viewport frequently and provide more 2D sketches for conveying their rough design ideas. 
Moreover, panel layouts are an important part of freeform architectural design that induces geometric constraints, but they are rarely considered in existing approaches.

\begin{figure}[t]
    \centering
    \frame{\includegraphics[width=0.48\columnwidth]{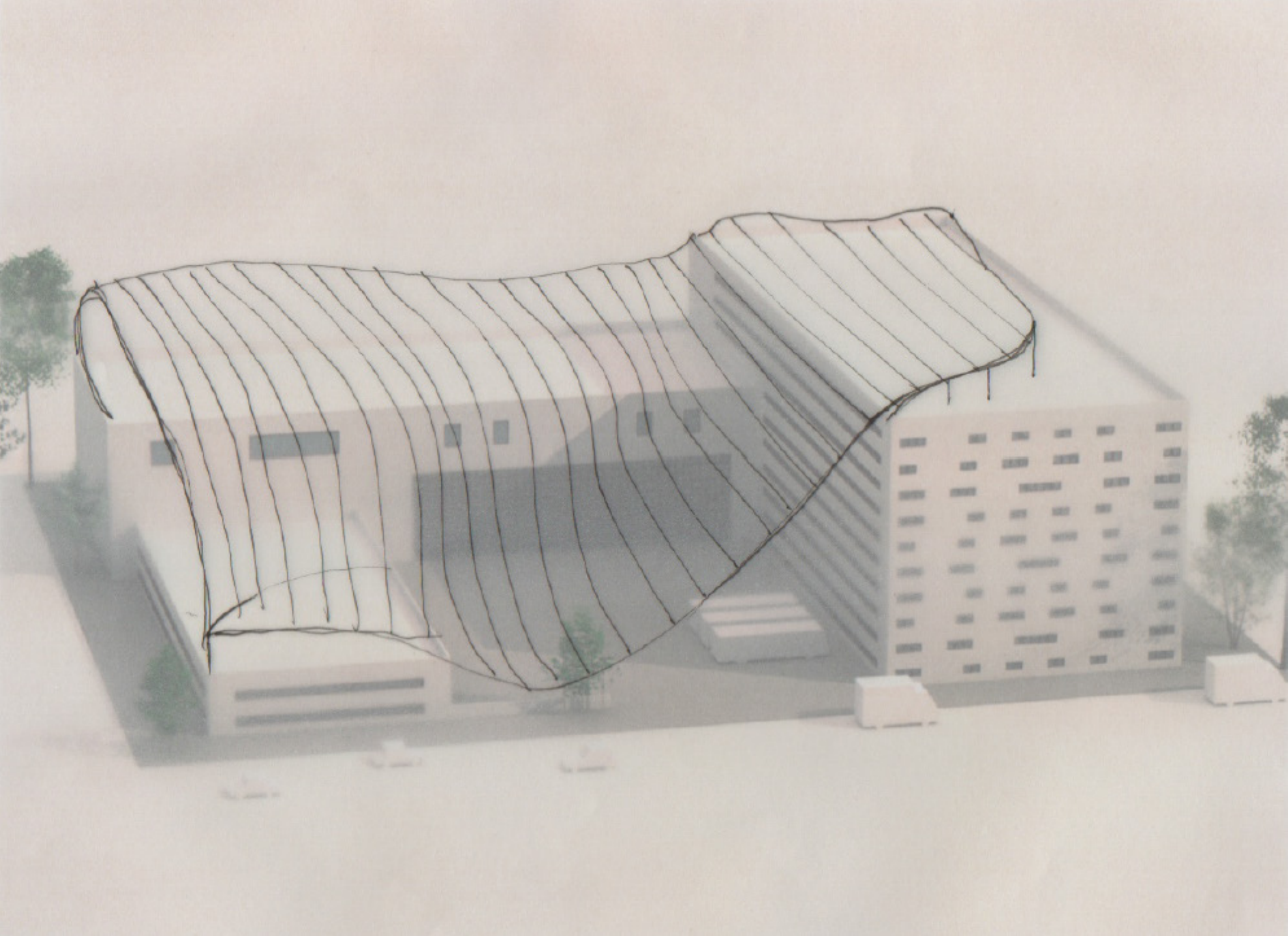}}
    \frame{\includegraphics[width=0.48\columnwidth]{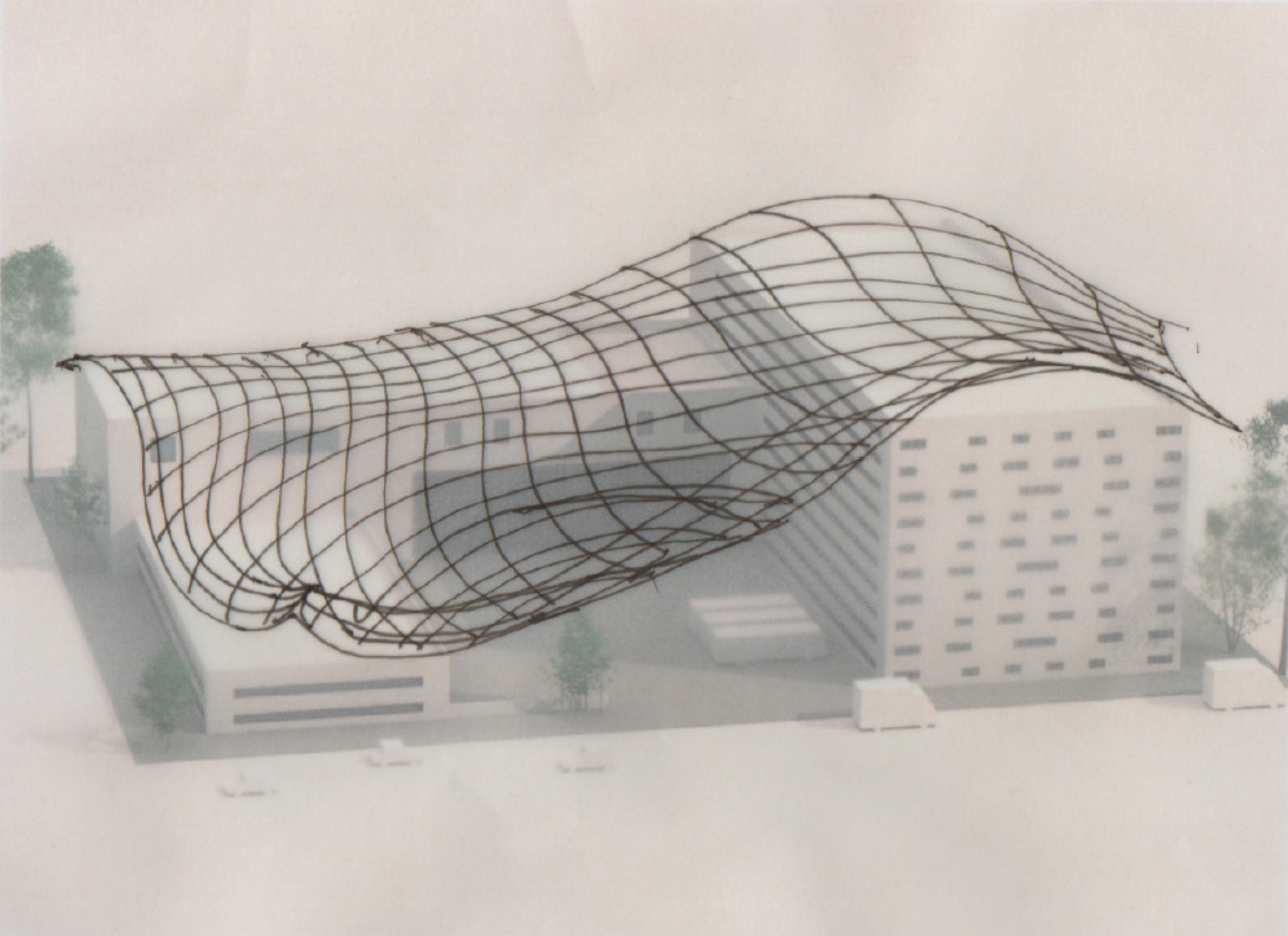}}
    \caption{Two different designs of freeform roofs via 2D sketching.}
    \label{fig:design}
\end{figure}

In this paper, we propose a sketch-based modeling system that allows users to design roof-like freeform architectural surfaces represented as planar quadrilateral (PQ) meshes. A PQ mesh represents a layout of flat panels with quadrilateral shapes, which is a preferable way to realize freeform glass structures due to their benefits in the construction cost~\cite{Glymph:2004}. To create such shapes, existing methods typically start from a reference surface that is modeled separately, and optimize a quadrilateral mesh to enforce face planarity while aligning with the reference shape~\cite{Liu:2006,Liu:2011}.
Our system enables a new modeling paradigm where the underlying surface shape and the PQ mesh layout are determined simultaneously from a single 2D sketch. 
To specify the surface structure, the user can draw the boundary lines and suggestive contour lines, as well as annotate the self-occluded parts. In addition, the user can draw sparse feature lines within the surface area to indicate the edge directions of the PQ mesh. 
To determine the mesh at an interactive rate, we train a deep neural network that takes the 2D sketch as input and first infers the depth and normal maps of the visible occluded regions, which are further fused via a network module to a smooth B-spline surface patch that represents the underlying surface shape. The network also infers a dense conjugate direction field (CDF) over the surface as an indication of the PQ mesh layout~\cite{Liu:2006}. 
The final PQ mesh is then extracted from the B-spline surface and the CDF.
To train and test our neural network, we propose a method to synthesize a large amount of data containing roof-like B-spline surfaces, their PQ mesh layouts, and the corresponding sketches.
We verify the effectiveness and usability of our system using a series of quantitative and qualitative evaluations, ablation studies, and user studies. To summarize, our contribution includes:
\begin{itemize}
    \item We develop a novel sketch-based modeling system for PQ meshes. Utilizing deep neural networks, our system can produce a PQ mesh from a single sketch at an interactive rate, providing an intuitive way to model the surface shape and the PQ mesh layout simultaneously.
    \item We propose a method to generate a large dataset of surfaces and sketches that is suitable for the training and testing of our network.
\end{itemize}

\begin{figure}[t]
    \centering
    \includegraphics[width=\linewidth]{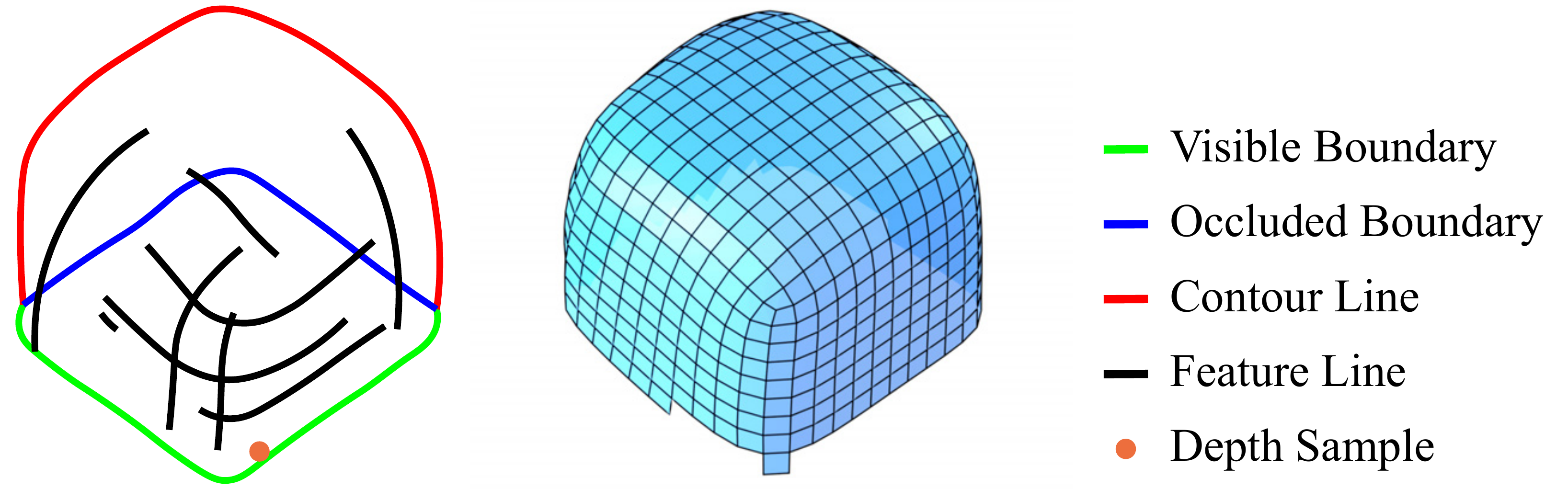}
    \caption{An example sketch using our system and the resulting PQ mesh.}
    \label{fig:Sketching_processing_diagram}
\end{figure}
 
\begin{figure*}[t]
 \includegraphics[width=\linewidth]{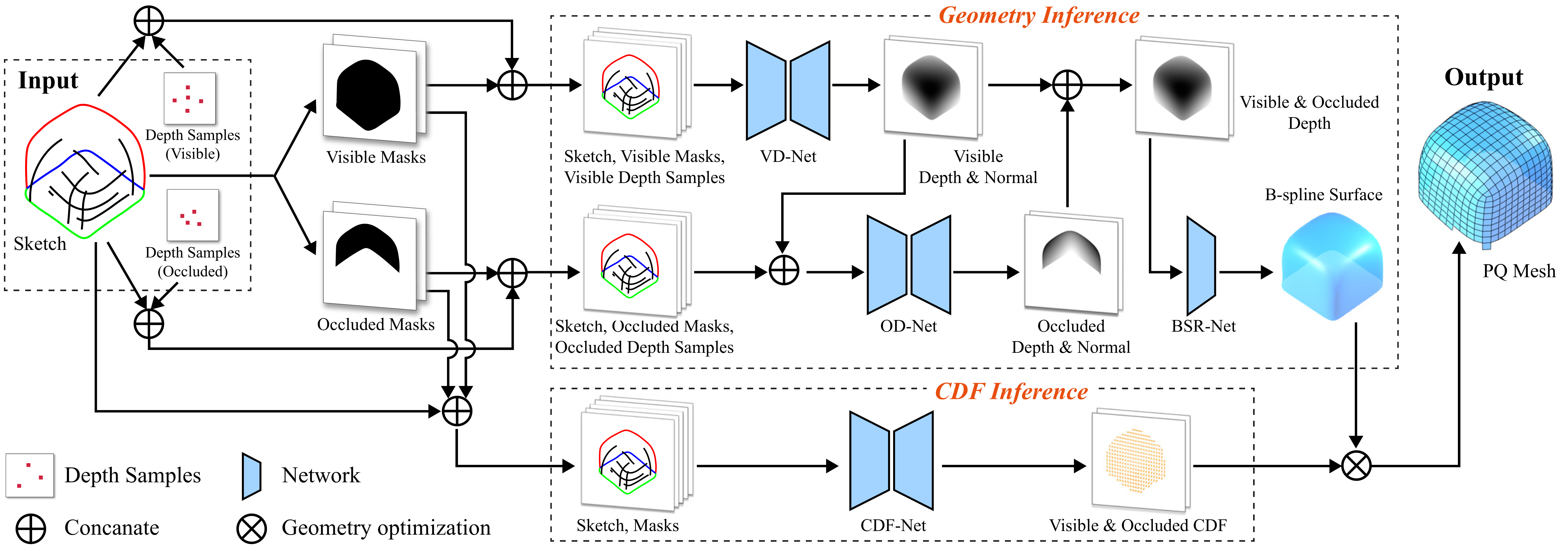}
 \caption{The algorithm pipeline of our Sketch2PQ system. The system takes stroke lines, depth samples, and the visible and occluded region masks induced from the sketch as input. 
 A geometry inference module predicts the depth and normal maps for the visible and occluded regions, and use them to infer a B-spline surface. 
 In addition, a conjugate direction field (CDF) inference module predicts a CDF that approximates the PQ mesh layout. Finally, a PQ mesh is extracted from the B-spline surface and the CDF via geometry optimization.}
 \label{fig:algorithm}
\end{figure*}

\section{Related Work}

In this section, we briefly review existing works closely related to our paper. More in-depth exposition of differential geometry and deep learning can be found in~\cite{doCarmo1976} and \cite{Goodfellow2016}.

\myparagraph{Sketch-based modeling} Sketch-based modeling has been studied extensively in the past. In the following, we focus on methods that create 3D shapes from 2D sketches without the need for an existing 3D model, as the same scenario is considered in this paper. The readers are referred to~\cite{Olsen:2009:survey,Cordier:2016:SBM} for more complete overviews that cover other scenarios.

Many existing works utilize geometry priors from human experience or specific object categories to determine the 3D shapes.
Igarashi \etal~\cite{Igarashi:1999:Teddy} utilized regular shape composition to model stuffed animals and other rotund objects. Tai \etal ~\cite{Tai:2004:ConvSurf} designed 3d models from sketched silhouette curves based on a convolution surface model. Schmidt \etal ~\cite{Schmidt:2005:BlobTree} increased the shape complexity by using Hierarchical Implicit Volume Models as an underlying shape representation and applying blending and CSG operations to implicit volumes inflated from 2D contours.
To model 3D trees, Chen \etal~\cite{Chen2008Tree} first constructed a 3D branch model using a database of trees as prior, and then added leaves based on botanical rules. Chen \etal~\cite{Chen2008} constructed an architectural model from a freehand sketch by identifying primitive geometries, detailed geometries, and textures using databases of such elements as priors.
Gingold \etal~\cite{Gingold2009Structured} developed a sketching system where the user can specify semantic information that helps to interpret the 2D sketch. Rivers \etal ~\cite{Rivers20103D} modeled 3D objects from multi-view silhouettes by constructing 3D parts from the silhouettes and assembling them using CSG operations. Entem \etal~\cite{Entem14} used smoothness and structural symmetry priors to model 3D animals from a side-view sketch.
Jung~\etal ~\cite{Jung:2015:SketchingFolds} inferred quasi-developable surfaces with pre-designed folds from multi-view sketches, by interleaving developability optimization with identification of silhouette points.
Dvoro\v{z}\v{n}\'{a}k \etal ~\cite{Dvoroznak2020Monster} combined 3D inflation with a rigidity-preserving and layered deformation model to produce a smooth 3D mesh from a single-view sketch.
Li \etal~\cite{BendSketch} utilized detailed curvature hints in the 2D sketch to infer local shapes and model freeform surfaces. Similarly, the sparse feature lines used in our system follow the PQ mesh layout and provide a strong hint for the local curvature.

Recently, learning-based approaches have become a promising way to handle different types of sketch inputs and predict 3D shapes robustly. 
Huang \etal~\cite{huang2016shape} and Nishida \etal ~\cite{Nishida2016Interactive} utilized  deep convolutional neural networks (CNN) to predict procedural modeling parameters from the input sketch. Han \etal~\cite{Han:2017:sketch2face} proposed a sketching system for 3D face and caricature modeling, using a CNN-based deep regression network to infer the coefficients for a bilinear face representation.
Lun \etal~\cite{Lun:2017:SketchModeling} proposed an encoder-decoder network to infer multi-view depth and normal maps from the input sketch, and consolidated them into a point cloud via optimization.
Su \etal~\cite{FuHongBo2018} inferred normal maps from 2D sketches by treating the process as an image translation problem and solving it using a  Generative Adversarial Network (GAN) framework. 
A GAN framework was also used in \cite{wang:2018a:generatemodel} to reconstruct 3D models from hand-drawn sketches.
Delanoy \etal~\cite{Delanoy20183D} proposed an encoder-decoder network to predict the occupancy of 3D voxels from a line drawing. 
To infer freeform surfaces from sparse sketches, Li \etal~\cite{Li:2018:SketchCNN} used CNNs to infer the depth and normal maps representing the surface, with an intermediate layer that models the curvature direction field and produces a confidence map to improve robustness.
In comparison, we also utilize neural networks to infer depth maps for the target surface, but with additional modules to infer conjugate direction fields that help to extract a PQ mesh.
Du \etal~\cite{Du2020SAniHead} created animal-like 3D head meshes from dual-view sketches, using Graph Convolutional Neural Networks to initialize and refine the model.
Yan \etal ~\cite{Yan2020Interactive} adopted a conditional generative adversarial network (CGAN) trained with physics-based simulation data to produce raw liquid splash models from input sketches, and further refined them to achieve realistic results.
Han \etal~\cite{Han2020Reconstructing} reconstructed a 3D model from multiple sketches by optimizing the 3D shape to match a set of attenuance images generated from the sketches using a CGAN.
Li \etal~\cite{Li:2020:Sketch2CAD} presented a CAD modeling system where a deep neural network infers shape edits from the user sketches to incremental create a 3D object.
Yang \etal~\cite{YANG2021Learning} proposed a two-stage learning framework for 3D face reconstruction from sketches, integrating the knowledge of face reconstruction from
photos. 
To generate 3D shapes from poorly-drawn sketches, Zhang \etal~\cite{Zhang_2021_CVPR} proposed an encoder-decoder architecture that disentangles the shape and the viewpoint in the latent space.
Smirnov \etal~\cite{smirnov2021patches} used a deformable template consisting of Coons patches to learn 3D man-made shapes from 2D sketches.

\myparagraph{PQ meshes}
Discretization of freeform surfaces using PQ meshes has been studied from both theoretical and practical perspectives. 
Liu \etal~\cite{Liu:2006} noted that PQ meshes are discrete counterparts of conjugate curve networks; they derived special PQ meshes including conical and circular meshes from networks of principal curvature lines, a special type of conjugate curve network.
Pottmann \etal ~\cite{Pottmann2007b} proposed methods to compute parallel meshes with planar faces, with applications for multiple-layer freeform structures. Zadravec \etal~\cite{zadravec-2010-vf} and Liu \etal~\cite{Liu:2011} designed general PQ meshes by computing smooth conjugate direction fields (CDF) on the reference surface. The CDF computation can also be done by optimizing an $N$-PolyVector field~\cite{ComplexRoots:Diamanti:2014} without the need for integer variables. All the methods above seek a PQ mesh that approximates an existing surface. Another way of PQ mesh design is to explore the shape space of PQ meshes with the same mesh topology but different vertex positions~\cite{Yang2011,Vaxman2012,Zhao2013Intuitive,Deng2013Exploring,Vaxman2014,DENG201513,Poranne2015}, starting from an existing PQ shape.
Neither approach is suitable for sketch-based modeling of PQ meshes, since they both require an initial 3D shape which is not available from the user input.

\section{The Sketch2PQ modeling system} \label{sec:system}
We develop a sketch-based modeling system for freeform PQ structures (see Fig.~\ref{fig:Sketching_processing_diagram} for its pipeline). The system supports four types of strokes for depicting the surface structure and the PQ mesh layout (Sec.~\ref{subsec:ui}). The sketch image is fed to a neural network to infer a B-spline surface for the underlying shape, as well as a conjugate direction field (CDF) that indicates the PQ mesh layout. The final PQ mesh is generated from the surface shape and the CDF (Sec.~\ref{subsec:network}).

\subsection{System Input} \label{subsec:ui}

\firstparagraph{Stroke types}
The input to our system is an architectural style sketch represented as an RGB image, depicting the axonometric projection of a roof-like surface with a boundary. Following the design habit of architectural users, we allow two classes of stroke lines (see Fig.~\ref{fig:Sketching_processing_diagram} for an example):
\begin{itemize}
 \item \emph{Structural lines} convey the structural information of the surface. There are three types of structural lines:
 \begin{itemize}
     \item A \emph{visible boundary} is a visible part of the surface boundary; it is shown in \textbf{{green}}.
     \item An \emph{occluded boundary} is an occluded part of the surface boundary; it is shown in \textbf{{blue}}.
     \item A \emph{contour line} is part of the boundary of the surface area on the canvas but is not part of the 3D surface boundary; it is shown in \textbf{{red}}.
 \end{itemize}
 \item \emph{Feature lines} are the direction lines that the resulting PQ mesh edges should follow; they are shown in \textbf{black}.
\end{itemize}
Similar to~\cite{BendSketch}, after the user draws a stroke line on the 2D canvas, it is vectorized to a cubic B-spline curve via curve approximation for easy editing. To reduce ambiguity, we assume that the features lines are drawn on visible surface regions only.
The vectorized sparse sketch is rendered on a 256$\times$256 RGB image with a line width of three pixels, and we apply a 3$\times$3 Gaussian filter on the original image.

\begin{wrapfigure}{r}{0.5\columnwidth}
  \centering
  \vspace*{-1.5ex}
  \includegraphics[width=0.49\columnwidth]{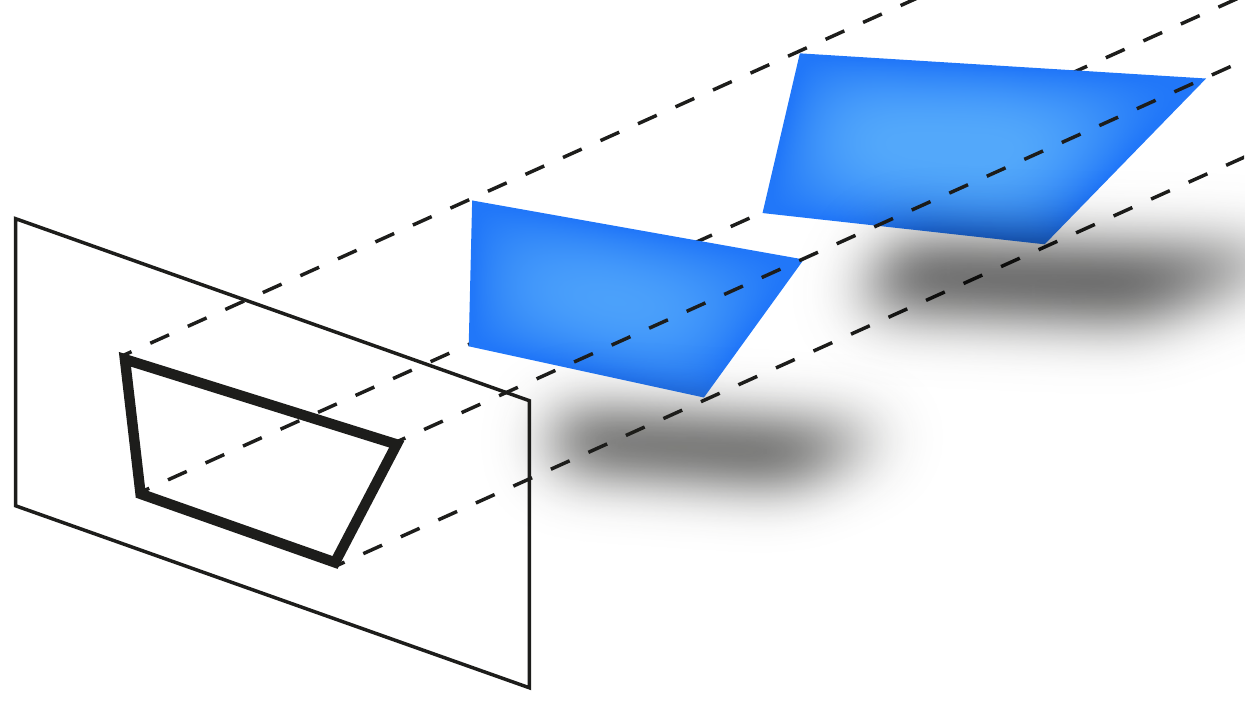}
  \vspace*{-2ex}
\end{wrapfigure}
The stroke lines roughly follow the boundaries, contours, and a subset of the interior edges of a PQ mesh. This can only determine the PQ mesh shape up to scaling and translation along the view direction. This is because for any quadrilateral with co-planar vertices in the 3D space, scaling and/or translation along 
the view direction will keep the vertices co-planar while producing the same projection onto the viewing plane 
(see the inset figure for an example). As a result, there is a family of PQ meshes with the same projection onto the viewing plane. 
Our system only returns one PQ mesh among the valid solutions. To enable more fine-grained control, we allow depth samples as optional inputs to provide depth cues at certain parts of the surface, similar to~\cite{Li:2018:SketchCNN}. 
As each depth sample can be specified for either a visible or occluded part of the surface, we store two depth sample maps for the visible parts and the occluded parts, respectively. Each is represented as a single-channel map, which contains the specified depth values for pixels corresponding to the sample points and zero values for other pixels (note that we assume the whole 3D surface have positive depth values).

\begin{figure}[t]
    \centering
    \includegraphics[width=0.7\columnwidth]{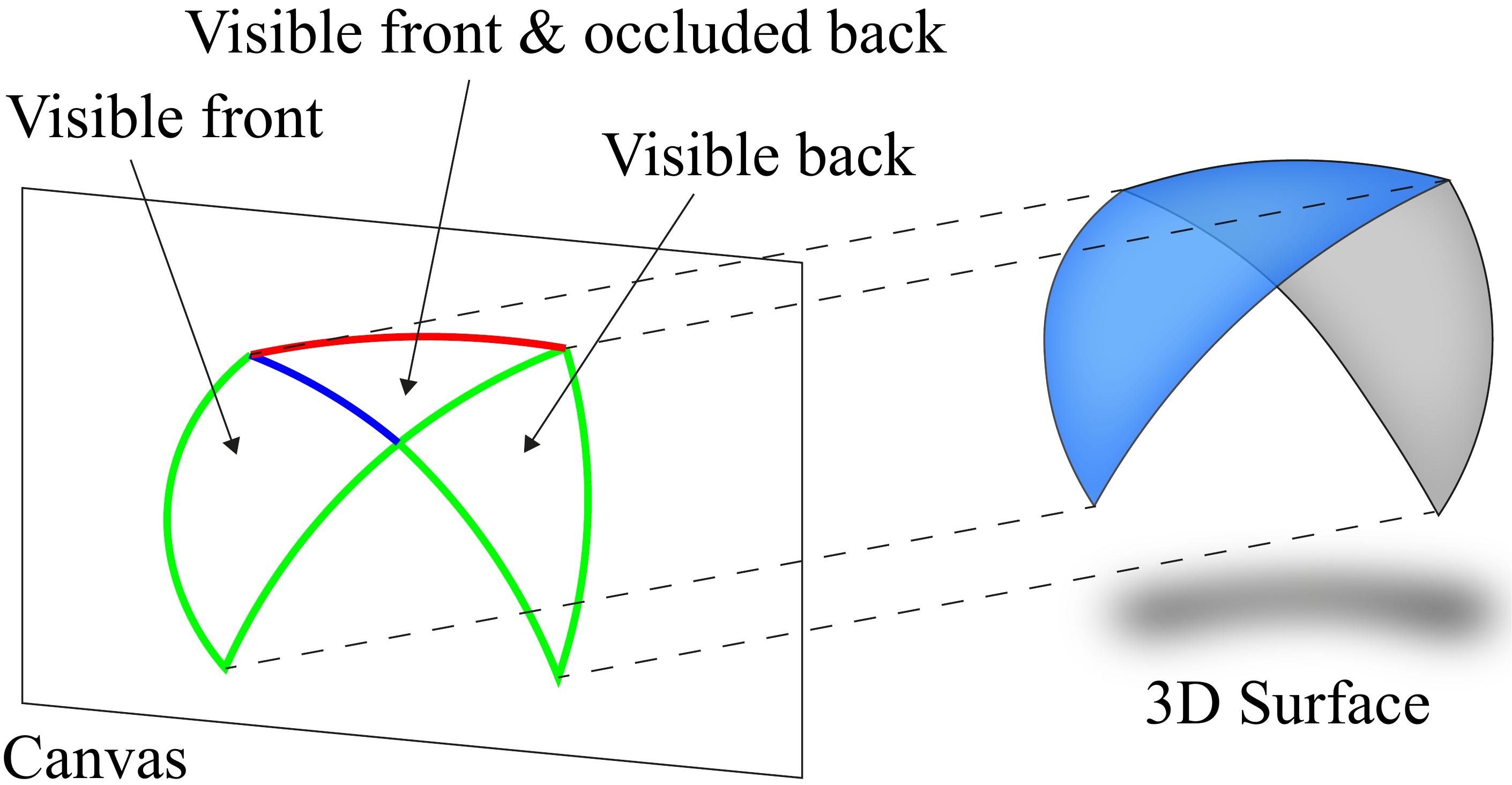}
    \caption{An example of visibility deduced from a sketch.}
    \label{fig:SketchVisibility}
\end{figure}

\myparagraph{Visibility and orientation masks}
The input sketch is first processed to derive visibility attributes that facilitate the inference of the surface shape. The attributes are fed into the network together with the user input.
We consider the sketch as a union of disjoint regions (see Fig.~\ref{fig:SketchVisibility}), where each region is enclosed by structural lines (i.e., visible/occluded boundaries and contour lines) and contains no other structural lines in the interior. We call such a region a \emph{minimal region}.
All minimal regions can be identified using a seeded region growing algorithm~\cite{Seed_Adams_1994}. 
For each minimal region, we determine: (1)~whether it corresponds to a visible or occluded part of the surface, and (2)~the orientation of the corresponding surface part, i.e., whether its side facing the canvas is the front side or the back side of the surface. 
To determine the properties, we make the following assumptions:
\begin{assumption}
\label{assump:SurfaceConditions}
The surface and the sketch satisfy:
\begin{enumerate}[label=(\alph*)]
    \item The surface is of disk topology, and has no self-intersection.
    \item The surface projects onto a disk-topology region in the canvas.
    \item Any straight line parallel to the view direction has no more than two intersection points with the surface.
    \item For a minimal region that corresponds to an occluded part of the surface, its enclosing structural lines must contain a contour line.
    \item Two structural lines can only overlap at isolated points.
    \item Among all visible parts of the surface, those with their front side facing the canvas have a larger total projected area on the canvas than those with their back side facing the canvas.
\end{enumerate}
\end{assumption}
These assumptions help to simplify the problem while covering a sufficiently large space of surface shapes.
Under the assumptions, each minimal region corresponds to a visible part of the surface. We determine the orientation of all visible parts based on the adjacency relation between their corresponding minimal regions and the following properties:
\begin{proposition}
\label{prop:Orientation}
The orientations of the visible parts satisfy:
\begin{enumerate}[label=(\alph*)]
    \item If two minimal regions share a visible boundary line, then their corresponding visible parts have opposite orientations.
    \item If two minimal regions share an occluded boundary line, then their corresponding visible parts have the same orientation.
\end{enumerate}
\end{proposition}
A proof is given in Appendix~\ref{appx:Proof}.
Accordingly, we first pick an arbitrary minimal region and set the orientation of its corresponding visible part to be front, then propagate the orientation information to all visible parts using Proposition~\ref{prop:Orientation}. We then compute the total areas of the minimal regions with front and back orientations, respectively. If the back regions have a larger total area, we negate the orientation of all visible parts according to Assumption~\ref{assump:SurfaceConditions}(f).
Finally, Assumption~\ref{assump:SurfaceConditions}(d) implies that each minimal region incident with a contour line corresponds to a visible part and an occluded part with opposite orientations, and we use this rule to identify all occluded parts and their orientations.

We divide all surface parts into four categories based on their visibility and orientation: \emph{visible-front}, \emph{visible-back}, \emph{occluded-front}, and \emph{occluded-back}. For each category, we construct a binary mask indicating all minimal regions that correspond to surface parts of the category (see Fig.~\ref{fig:VisibilityOrientation} for some examples).
These  masks are fed to the neural networks along with the input sketch and depth sample maps.
They are also used to post-process the output of the networks, by discarding information that is outside the surface region. 

\begin{figure}[t]
    \centering
    \includegraphics[width=\columnwidth]{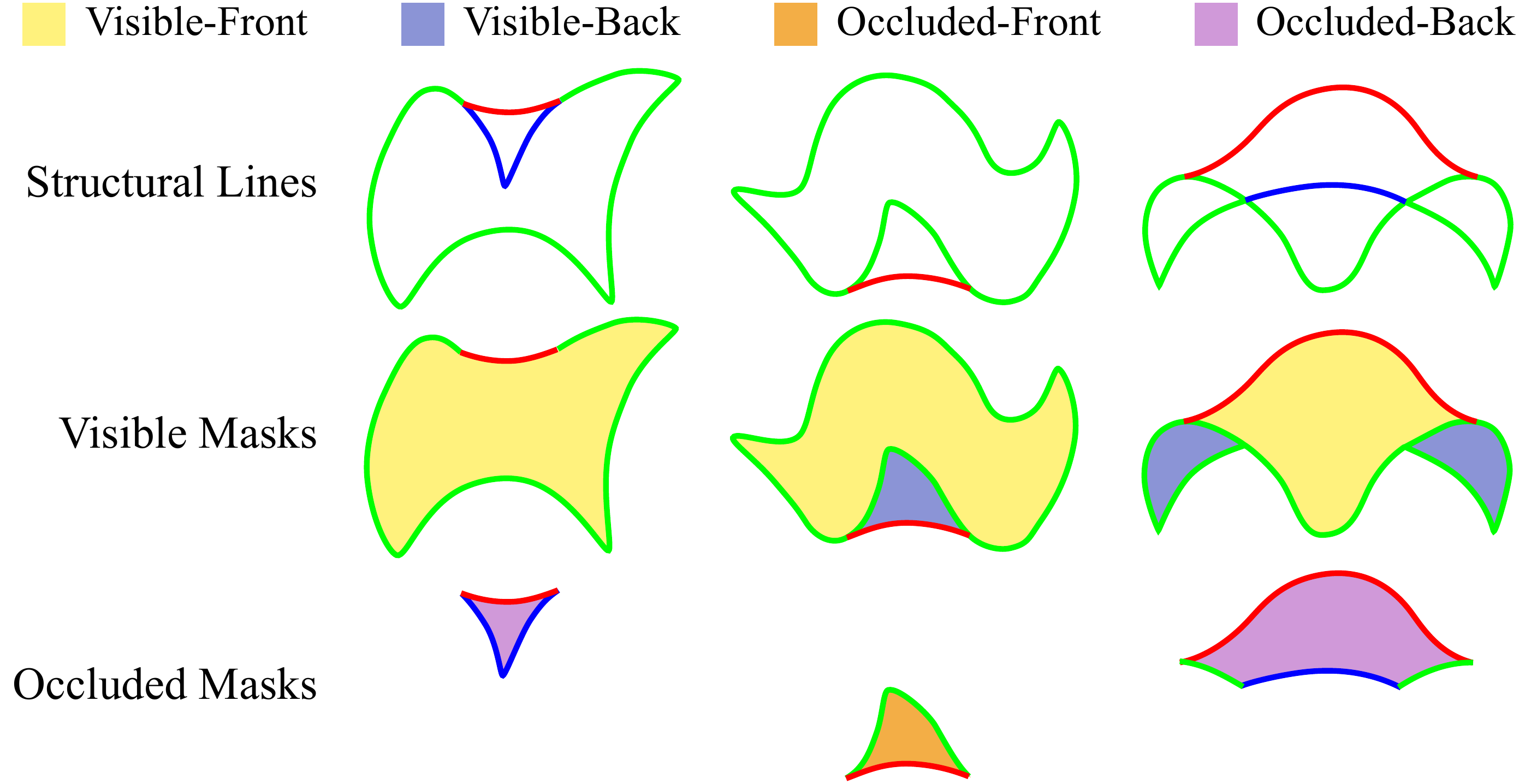}
    \caption{Examples of visibility-orientation masks for different sketches. Front and back masks of the same visibility are shown together.}
    \label{fig:VisibilityOrientation}
\end{figure}

\subsection{Network Structure and Loss Functions} \label{subsec:network}
The user input and  the visibility-orientation masks obtained in Sec.~\ref{subsec:ui} are used to derive a PQ mesh. As the problem is highly non-linear, direct inference with a single network often produces unsatisfactory results. We instead utilize multiple modules to infer different shape attributes that determine the final PQ mesh.
Specifically, we first use two network modules to predict the depth and normal maps for the visible and occluded parts of the surface, respectively (Secs.~\ref{subsubsec:vd} \& \ref{subsubsec:od}). The depths are fed to another network module to infer the underlying surface shape represented as a B-spline surface (Sec.~\ref{subsubsec:surf}). In addition, a network module predicts a conjugate direction field (CDF) that serves as an indication of the PQ mesh layout~\cite{Liu:2006} (Sec.~\ref{subsubsec:cdf}). Finally, a geometry optimization is performed to compute the final PQ mesh according to the  B-spline surface and the CDF (Sec.~\ref{subsubsec:pq}).

\subsubsection{Depth and normal prediction for visible parts} \label{subsubsec:vd}
Given the input sketch, the visible-front and visible-back masks, and the visible depth samples, we adopt a U-Net like network module called \emph{VD-Net} (see Fig.~\ref{fig:Net_work_depth_normal}) to recover the depth and outward normal maps over the 2D canvas domain for the visible surface parts, similar to~\cite{Li:2018:SketchCNN}. Its input contains six channels: three for the RGB sketch image, two for the visible-front and visible-back masks, and one for the (optional) depth samples. Our VD-Net has five domain resolutions, from 256$\times$256 to 16$\times$16. The output is in eight channels: two for the depth maps of the visible-front and visible-back surface parts, and six for their normal maps.
The network is trained using the following loss function:
\begin{equation}
 \loss{v} = \lvisdat{} + \wvissmooth{} \lvissmooth{}  + \wviscons{} \lviscons{} + \wvisdepth{} \lvisdepth{},
 \label{eq:VDLoss}
\end{equation}
where the terms $\lvisdat{}, \lvissmooth{}, \lviscons{}, \lvisdepth{}$ enforce different properties as explained below. Their weights are empirically chosen as $\wvissmooth{}=0.2$, $\wviscons{} = 0.1$ and $\wvisdepth{} = 10$.

\myparagraph{Data term} The term $\lvisdat{}$ penalizes the $L_2$ error between the predicted depth/normal maps and their ground-truth counterparts on the visible surface parts:
\begin{equation}
 \lvisdat{} = \frac{1}{|\vispixset{}|} \sum\nolimits_{i \in \vispixset{}} \left( \vddatweightdepth{} (d_i-\gt{d_i})^2 + \vddatweightnormal{} \| \mn_i - \gt{\mn_i}\|^2 \right), 
  \label{eq:VDDataTerm}
\end{equation}
where $\vispixset{}$ denotes the set of pixels for the visible surface parts, and $d_i,\gt{d_i} \in \mathbb{R}$, $\mn_i,\gt{\mn_i} \in \mathbb{R}^3$ are the predicted and ground-truth depth and normal at a pixel $i$, respectively. The weights $\vddatweightdepth{}, \vddatweightnormal{}$ are set to $\vddatweightdepth{} = 1.0, \vddatweightnormal{} = 0.1$.

\myparagraph{Smoothness term} $\lvissmooth{}$ is a smoothness regularization term for the depth and normal maps, which  penalizes the difference between the values in neighboring pixels:
\begin{equation}
\lvissmooth{} = \frac{1}{|\visnbset{}|}\sum_{(i,j) \in \visnbset}\left(\vdsmoweightdepth{} (d_i - d_j)^2 + \vdsmoweightnormal{} \|\mn_i - \mn_j\|^2\right).
\end{equation}
Here $\visnbset{}$ is the set of neighboring visible pixels with the same orientation, and the weights are set to $\vdsmoweightdepth{} = 0.01, \vdsmoweightnormal{} = 0.001$.

\myparagraph{Depth-normal compatibility term} The predicted normal vectors and depth values should be compatible with each other. This is enforced via the following term: 
\begin{equation*}
\lviscons{} = \frac{1}{|\visxnbset{}|} \sum\nolimits_{i \in \visxnbset{}} (\mn_i \cdot \mt_{i,x})^2 + \frac{1}{|\visynbset{}|} \sum\nolimits_{i \in \visynbset{}} (\mn_i \cdot \mt_{i,y})^2,
%\label{eq:VisConsistency}
\end{equation*}
where $\visxnbset{}$ is the set of visible pixels where their right neighbor pixels are also visible with the same orientation, i.e., $\visxnbset{} = \{i \mid (i, \rightneighbor{i}) \in \visnbset{}\}$ where $\rightneighbor{\cdot}$ denotes the right neighbor pixel. Similarly,~$\visynbset{} = \{i  \mid (i, \upneighbor{i}) \in \visnbset{}\}$ where $\upneighbor{\cdot}$ denotes the upper neighbor pixel.
$\mt_{i,x} = (1, 0, (d_{\rightneighbor{i}} -d_{i})/\pixwidth{})$, 
$\mt_{i,y} = (1, 0, (d_{\upneighbor{i}} -d_{i})/\pixwidth{})$ are tangent vectors at pixel $i$ computed by finite difference, where $\pixwidth{}$ is the pixel width when mapped to the canonical scale of the training data. $\lviscons{}$ enforces orthogonality between the predicted normals and the tangent vectors derived from the predicted depths.

\myparagraph{Depth sample term}  Finally, $\vissamplepixset{}$ requires the predicted depth to be close to the provided depth samples:
\begin{equation}
 \lvisdepth{} = 
 \begin{cases}
 \frac{1}{|\vissamplepixset{}|}\sum_{i \in \vissamplepixset{}} (d_i - \spd{d_i})^2, & \text{if}~\vissamplepixset{} \neq \emptyset,\\
 0 & \text{otherwise}.
 \end{cases}
\end{equation}
where $\vissamplepixset{}$ is the set of visible pixels with depth samples, and $\spd{d_i}$ is the provided depth value for pixel $i$.

\begin{figure}[t]
 \includegraphics[width=\linewidth]{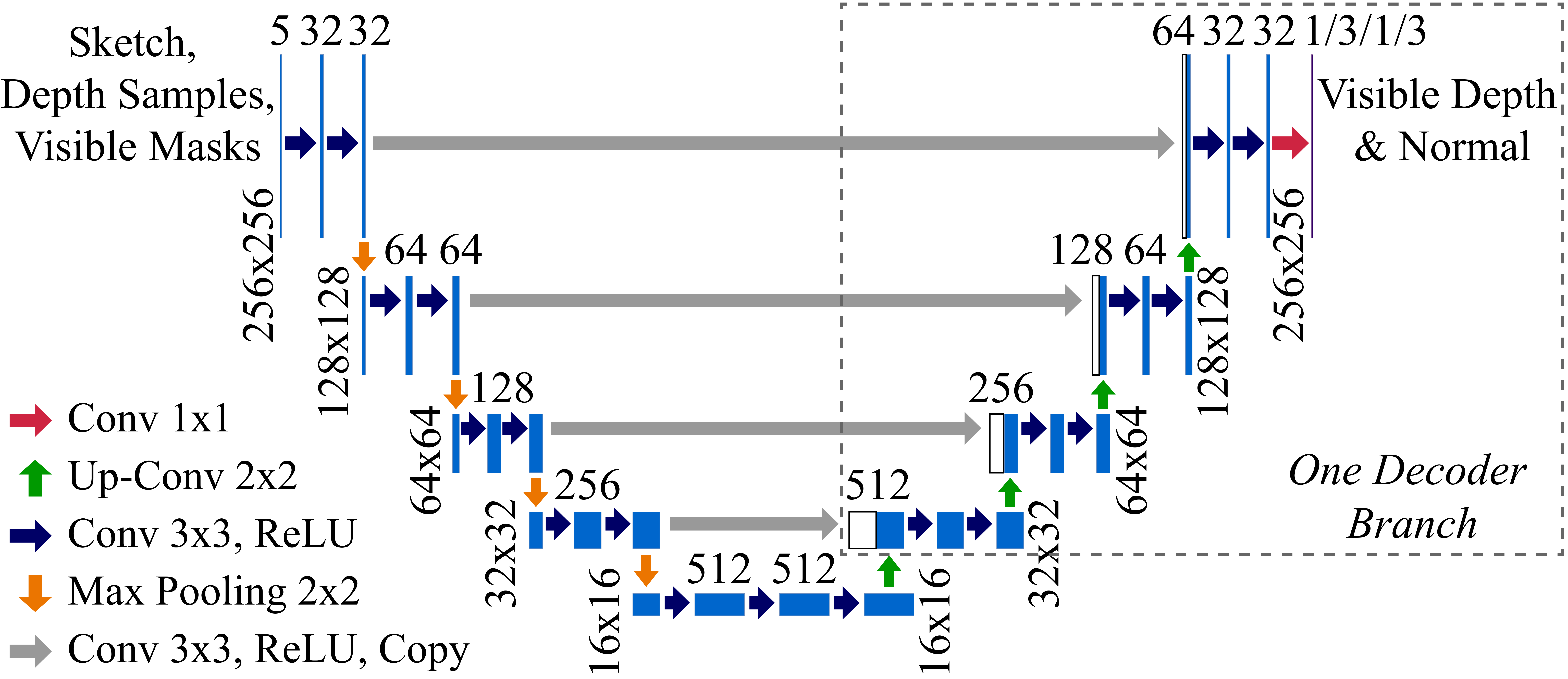}
 \caption{The VD-Net is an encoder-decoder network with five domain resolutions. The input is a 6-channel image containing the sketch, the visible-front/back masks and the visible depth samples. Four decoder branches share the same encoder and predict the depth \& normal maps for the visible front and back regions, respectively. }
 \label{fig:Net_work_depth_normal}
\end{figure}

\subsubsection{Depth and normal prediction for occluded parts} \label{subsubsec:od}
For the occluded parts, a setup similar to Sec.~\ref{subsubsec:vd} may be used to infer the depth and outward normal maps. However, the results may be inconsistent with their counterparts for the visible regions, i.e., the depth and normal functions may be discontinuous across the contour lines.
Therefore, we adapt the network structure and loss function to enforce the consistency. Specifically, we introduce a network structure, called \emph{OD-Net}, with its input being the sketch, the occluded-front and occluded-back masks, the occluded depth samples, and the depth and normal maps produced by the VD-Net.
The remaining parts of the network are the same as the VD-Net. We train the network using the following loss function:
\begin{equation}
\begin{aligned}
 \loss{o} = & ~\loccdat{} + \woccsmooth{} \loccsmooth{}  + \wocccons{} \locccons{}\\
 &~+ \woccdepth{} \loccdepth{}
 + \wgc{} \lgc{}.
\end{aligned}
 \label{eq:ODLoss}
\end{equation}
Here $\loccdat{}$,  $\loccsmooth{}$, $\locccons{}$, $\loccdepth{}$ are the data term, smoothness term, depth-normal compatibility term, and depth sample term for the occluded parts respectively, and are defined for the set $\occpixset{}$ of occluded pixels in a way similar to Sec.~\ref{subsubsec:vd}.  $\lgc{}$ is a geometric consistency term between the visible and occluded regions near the contour line:
\begin{equation}
\begin{aligned}
\lgc{} = \frac{1}{|\Omega|} \sum_{(i,j) \in \Omega} & \left[\gcweightdepth{} ( (d_i^{\text{v}} - d_j^{\text{o}})^2 + (d_i^{\text{o}} - d_j^{\text{v}})^2)\right. \\
& \left.~+  \gcweightnormal{} (\|\mn_i^{\text{v}} - \mn_j^{\text{o}}\|^2 + \|\mn_i^{\text{o}} - \mn_j^{\text{v}}\|^2)\right].
\end{aligned}
 \label{eq:GeomCons}
\end{equation}
Here $\Omega = \{(i,j) \mid i, j \in \mathcal{C} \cap \vispixset{} \cap \occpixset{}, j \in \fourneighborset{i} \}$, where $\mathcal{C}$ is the set of pixels that are no more than four pixels away from the contour line, and  $\fourneighborset{i}$ denotes a neighborhood of pixel $i$ which consists of $i$ itself as well as its upper neighbor, right neighbor, and upper-right neighbor. The symbols $d^{\text{v}}, \mn^{\text{v}}$ and $d^{\text{o}}, \mn^{\text{o}}$ denote the depth and normal for the visible and occluded regions, respectively. The weights in Eq.~\eqref{eq:GeomCons} are set to $\gcweightdepth{}=1, \gcweightnormal{}=0.1$, and the weights in Eq.~\eqref{eq:ODLoss} are set to $\woccsmooth{}=0.2,\wocccons{}=0.1,\woccdepth{}=10, \wgc{}=0.001$.

\subsubsection{Surface prediction and refinement} \label{subsubsec:surf}
With the predicted depth maps for the visible and occluded parts, we reconstruct a cubic B-spline surface that represents the underlying shape of the PQ mesh. We use a B-spline representation because its smooth shapes are suitable for freeform architectural design. Moreover, it is compatible with NURBS-based 3D modeling software such as Rhino3D, which is popular among architects. This allows the design to be easily integrated into existing architectural design workflows and further modified in a later stage.
Therefore, even though it may be possible to infer the PQ mesh without the intermediate B-spline shape, we choose to reconstruct a B-spline surface to enable interoperability with the existing design pipeline. We use 30$\times$30 B-spline control points, which are sufficient for representing a large variety of shapes.
To determine the surface, one possibility is to optimize the control points to fit the depth maps. However, this is a non-convex problem due to the unknown correspondence between the depth map and the B-spline surface, and a good fitting would require proper initialization which is not a trivial problem. 
We instead use a network module, called \emph{BSR-Net}, to predict the control points $\{\cpoint{i,j} \in \mathbb{R}^3 \mid 1 \leq i, j  \leq 30\}$
from the depth maps produced by VD-Net and OD-Net.
We choose ResNet-18~\cite{He2015} as its structure and train it using a loss function:
\begin{equation}
 \lbsr{} = \lbsrdat{} + \wbsrfair{} \lbsrfair{}.
\end{equation}
Here $\lbsrdat{}$ is a data term that measures the difference between the predicted and ground-truth B-spline surfaces:
\begin{equation}
    \lbsrdat{} = \frac{1}{|\bsrparamset|} \sum\nolimits_{(u,v) \in \bsrparamset}
    \|\bsrsurf{u,v} - \bsrsurfgt{u,v}\|^2,
\end{equation}
where $\bsrparamset{}$ is a set of 100$\times$100  parameters sampled regularly from the parameter domain of the surface, and $\bsrsurf{u,v}$, $\bsrsurfgt{u,v}$ are the surface points at the parameter $(u, v)$ using the predicted and ground-truth control points, respectively. 
$\lbsrfair{}$ is a fairness term for the predicted control net:
\begin{equation}
    \lbsrfair{} = 
    \sum\nolimits_{i=2}^{29} \sum\nolimits_{j=2}^{29}
    \| \ulaplacian{i,j} \|^2 + \| \vlaplacian{i,j} \|^2,
    \label{eq:bsrfair}
\end{equation}
where $\ulaplacian{i,j}$ and $\vlaplacian{i,j}$ are the polyline Laplacians for the two control polygons at $\cpoint{i,j}$:
\begin{equation*}
    \ulaplacian{i,j} =  \cpoint{i,j} - \frac{\cpoint{i-1,j} + \cpoint{i+1,j}}{2},~~ 
     \vlaplacian{i,j} = \cpoint{i,j} - \frac{\cpoint{i,j-1} + \cpoint{i,j+1}}{2}.
\end{equation*}
The weight $\wbsrfair{}$ is set to $0.1$. 
During the test phase, we use the predicted surface as initialization and optimize its control points to further align it with the depth maps:
\begin{equation}
    \min_{\{\cpoint{i,j}\}}
    ~~
    \frac{1}{|\dpointset{}|} \sum\nolimits_{\dpoint{i} \in \dpointset{}}
    \left(\text{dist}(\dpoint{i})\right)^2 + 
    \wsurffair{} \lbsrfair{},
\end{equation}
where $\dpointset{}$ is the set of target 3D points derived from the input depth maps, $\text{dist}(\dpoint{i})$ is the distance from a point $\dpoint{i}$ to the surface, and $\wsurffair{} = 2$. 
The problem is solved using an iterative B-spline fitting method~\cite{Farin:CAGD:2000}.
As the initial surface is close to the target points, typically a few iterations are sufficient.

\subsubsection{Conjugate direction field prediction} \label{subsubsec:cdf}
The feature lines in the sketch indicate the edge directions of the PQ mesh layout. To determine the PQ mesh, we need to first compute a dense conjugate direction field (CDF) on the surface that aligns with the feature lines. A CDF is a general cross field consisting of four directions for each point, which correspond to the edges of infinitesimally small planar quadrilaterals~\cite{Liu:2006}. As such, it is a suitable approximation of a PQ mesh layout and commonly used as initialization for PQ mesh generation~\cite{Liu:2011}. Existing methods for CDF computation need to solve a non-linear and non-convex optimization problem with a large number of variables~\cite{Liu:2011,ComplexRoots:Diamanti:2014}, which is too slow for our interactive system. Therefore, we design a network module, called \emph{CDF-Net}, to efficiently predict the CDF. Since the CDF directions are 3D vectors in the tangent plane at each point of the surface, it is sufficient to predict their 2D projections onto the canvas plane. The predicted 2D vectors can then be projected back into the corresponding tangent planes on the B-spline surface to obtain the 3D CDF directions.
To this end, our CDF-Net predicts a dense CDF for the visible regions and the occluded regions respectively.

\begin{figure}[t]
 \includegraphics[width=\linewidth]{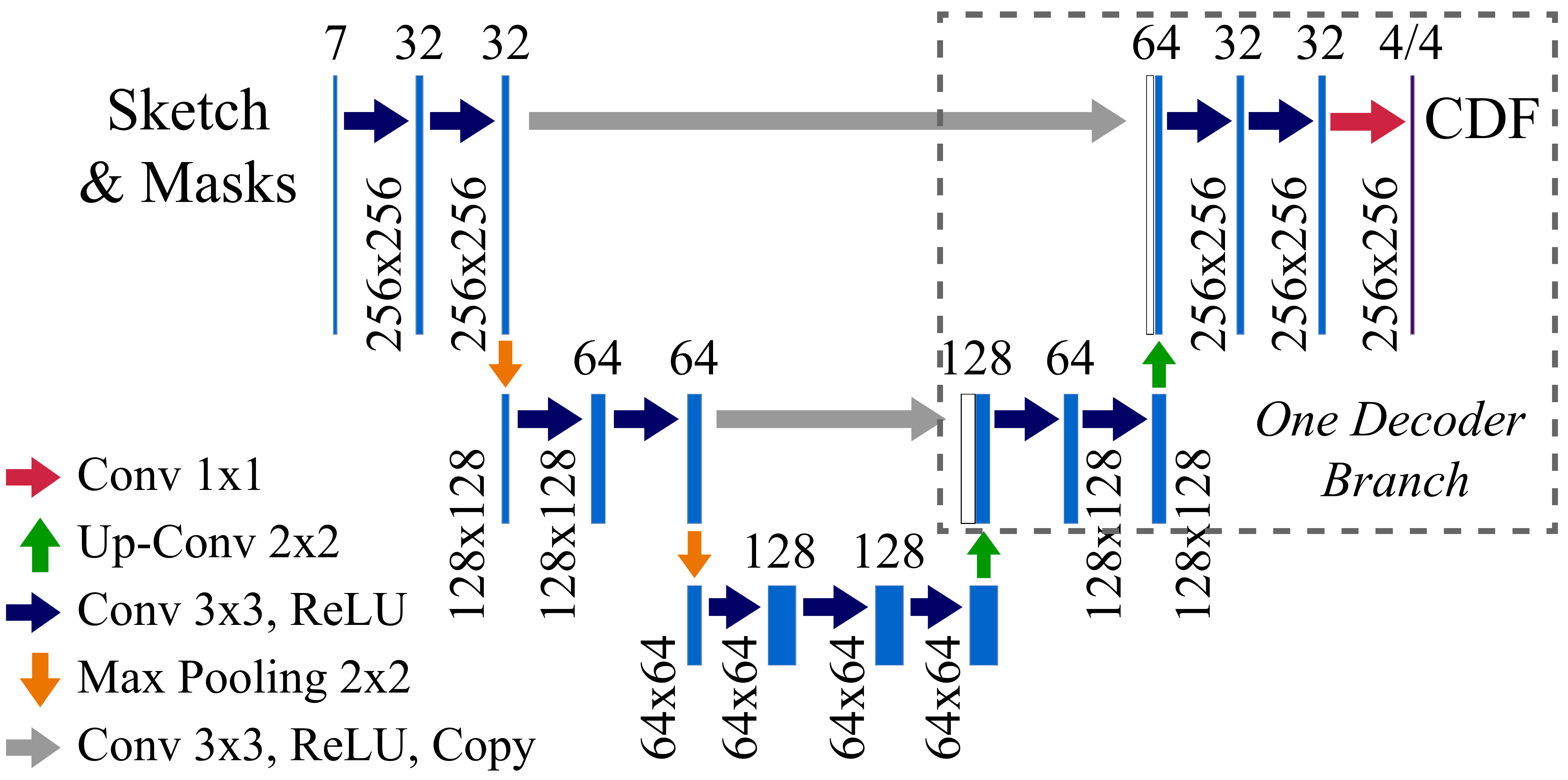}
 \caption{The CDF-Net is an encoder-decoder network with three domain resolutions. The input is a 7-channel image for the sketch and its four visibility-orientation masks. Two decoder branches share the same encoder and predict the CDF for the visible and occluded regions.}
 \label{fig:Net_field}
\end{figure}

The CDF-Net has a U-Net structure (see Fig.~\ref{fig:Net_field}), with a 7-channel input that includes the RGB sketch image and its four visibility-orientation masks. 
The output is a dense map with four channels, with a value $(c_{j,0}, c_{j,1}, c_{j,2}, c_{j,3})$ at each pixel $j$ to represent two vectors $\mathbf{u}_j = (c_{j,0}, c_{j,1})$ and $\mathbf{v}_j = (c_{j,2}, c_{j,3})$ that form the predicted projected CDF $\{\mathbf{u}_j, -\mathbf{u}_j, \mathbf{v}_j, -\mathbf{v}_j\}$. 
The network is trained to produce a prediction that is close to the ground-truth vectors $\{\mathbf{a}_j, - \mathbf{a}_j, \mathbf{b}_j, -\mathbf{b}_j \}$ at each pixel $j$. However, since each quadruplet of vectors is treated as an un-ordered set, there is an ambiguity in their correspondence when evaluating the difference between the predicted and ground-truth quadruplets. A similar ambiguity arises if we compare the predicted quadruplets at two neighboring pixels to evaluate the smoothness of the CDF.  
To avoid such issues, we follow~\cite{Li:2018:SketchCNN} and encode each quadruplet of vectors using their 4-PolyVector representation~\cite{ComplexRoots:Diamanti:2014}. Specifically, we treat the predicted vectors as complex numbers, i.e., $\compx{\mathbf{u}}_j = c_{j,0} + i c_{j,1}, \compx{\mathbf{v}}_j = c_{j,2} + i c_{j,3}$. Then the quadruplet $\{\compx{\mathbf{u}}_j, - \compx{\mathbf{u}}_j, \compx{\mathbf{v}}_j, -\compx{\mathbf{v}}_j\}$ are the roots of a complex polynomial $P(z) = (z - \compx{\mathbf{u}}_j) (z + \compx{\mathbf{u}}_j) (z - \compx{\mathbf{v}}_j) (z + \compx{\mathbf{v}}_j) = z^4 - ({\compx{\mathbf{u}}}_j^2 + {\compx{\mathbf{v}}}_j^2) z^2 + {\compx{\mathbf{u}}}_j^2 {\compx{\mathbf{v}}}_j^2$. We then identify the triplet with the polynomial coefficients ${\compx{\mathbf{u}}}_j^2 + {\compx{\mathbf{v}}}_j^2 = \compxenc{c}_{j,0} + i \compxenc{c}_{j,1}$ and ${\compx{\mathbf{u}}}_j^2 {\compx{\mathbf{v}}}_j^2 = \compxenc{c}_{j,2} + i \compxenc{c}_{j,3}$, and encode them as $\compxenc{\mathbf{c}}_j = (\compxenc{c}_{j,0}, \compxenc{c}_{j,1}, \compxenc{c}_{j,2}, \compxenc{c}_{j,3})$. The ground-truth vectors are encoded in the same way. We then evaluate the difference between two quadruplets by comparing their encoding vectors in $\mathbb{R}^4$. Based  on this encoding scheme, we train the CDF-Net with the following loss function:
\begin{equation}
\lcdf{} = \lcdfdat{} + \wcdfsmooth{} \lcdfsmooth{} + \wcdfcons{} \lcdfcons{}.
\label{eq:CDFLoss}
\end{equation}
Here $\lcdfdat{}$ is a data term that penalizes the difference between the predicted CDF and the ground-truth:
\begin{equation*}
  \lcdfdat{} = 
  \frac{1}{|\vispixset{}|} \sum_{j\in \vispixset{}}\| \compxenc{\mathbf{c}}_{j}^\text{v} - \compxenc{\mathbf{c}}_{j}^\text{v,gt} \|^2
  + 
  \frac{1}{|\occpixset{}|} \sum_{j\in \occpixset{}}\| \compxenc{\mathbf{c}}_{j}^\text{o} - \compxenc{\mathbf{c}}_{j}^\text{o,gt} \|^2
\end{equation*}
where $\compxenc{\mathbf{c}}_{j}^\text{v}$ and $\compxenc{\mathbf{c}}_{j}^\text{o}$ are the CDF encoding vectors at pixel $j$ for the visible and occluded region respectively, and  $\compxenc{\mathbf{c}}_{j}^\text{v,gt}, \compxenc{\mathbf{c}}_{j}^\text{o,gt}$ are the ground-truth encoding vectors. The term $\lcdfsmooth{}$ evaluates the smoothness of the predicted CDF by comparing its encoding vectors at neighboring pixels:
\begin{equation*}
\lcdfsmooth{} = \frac{1}{|\visnbset|} \sum_{(j,k) \in \visnbset{}} \| \compxenc{\mathbf{c}}_{j}^\text{v} - \compxenc{\mathbf{c}}_{k}^\text{v} \|^2
+ 
\frac{1}{|\occnbset|} \sum_{(j,k) \in \occnbset{}} \| \compxenc{\mathbf{c}}_{j}^\text{o} - \compxenc{\mathbf{c}}_{k}^\text{o} \|^2,
\end{equation*}
where $\visnbset{}$ and $\occnbset{}$ are the sets of neighboring pixel pairs in the visible and occluded regions, respectively.
The term $\lcdfcons{}$ enforces consistency of the CDF across the contour line, similar to the term $\lgc{}$ as in~\eqref{eq:GeomCons}:
\begin{equation}
\lcdfcons{} = \frac{1}{|\Omega|}\sum\nolimits_{ (j,k) \in \Omega } (\|\compxenc{\mathbf{c}}_{j}^\text{o} - \compxenc{\mathbf{c}}_{k}^\text{v} \|^2
+ \|\compxenc{\mathbf{c}}_{j}^\text{v} - \compxenc{\mathbf{c}}_{k}^\text{o} \|^2),
\end{equation}
where the set $\Omega$ is the same as in Eq.~\eqref{eq:GeomCons}.
The weights in Eq.~\eqref{eq:CDFLoss} are set to $\wcdfsmooth{}=0.1$ and $\wcdfcons{}=0.01$.

\subsubsection{PQ mesh generation} \label{subsubsec:pq}
Once the projected CDF and the B-spline surface patch are obtained, we are ready to generate a planar quadrilateral mesh that follows the CDF as much as possible. We first discretize the B-spline surface into a triangle mesh according to pixel resolution via depth buffer rasterization, and project the CDF along the view direction onto the mesh. Then we utilize a field-based quad meshing algorithm~\cite{Liu:2011} 
to extract a quadrilateral mesh that is approximately planar, and further optimize the face planarity using the method from~\cite{ShapeOp_2015}.
Since the CDF already provides a good approximation of the edge directions in a PQ mesh, the final PQ mesh typically aligns well with the CDF and the sketched feature lines (see Fig.~\ref{fig:Sketch_CDF}).

\begin{figure}[t]
 \includegraphics[width=\linewidth]{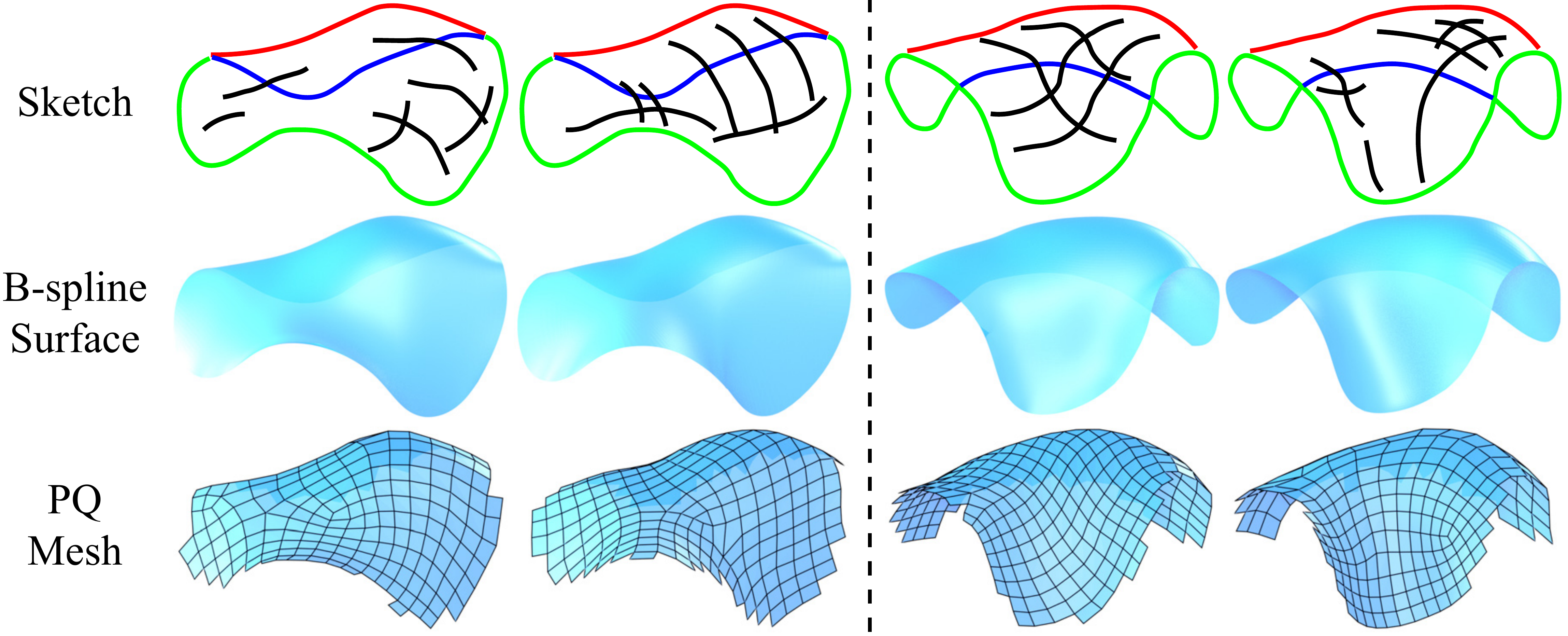}
 \caption{Feature lines can control the PQ mesh layout. Here we show two pairs of sketches, each having the same structural lines but different feature lines. The PQ mesh edges align well with the feature lines.}
 \label{fig:Sketch_CDF}
\end{figure}
\section{Data Generation and Network Training}
\label{sec:data_training}

The training data of our network needs to contain the ground-truth B-spline surfaces, their CDFs, and compatible sketch images. 
To the best of our knowledge, there is no large-scale 3D model dataset for freeform architectural shapes. Therefore, we build our own  dataset for training and testing.

\myparagraph{Shape creation} \label{subsec:3d_shape_generate}
Since our system is aimed at roof-like structures, we assume the sketched shape to be a height-field surface $z = \heightfield{x,y}$. Therefore, we create each ground-truth surface by first randomly generating a height-field B-Spline surface, and then applying a random 3D deformation to introduce variations. Specifically, we generate the control points $\{\cpoint{i, j}$ ($1 \leq i,j \leq 30$)$\}$ in the following steps.

First, we generate the $(x,y)$-coordinates $\cpxycoord{i,j}$ for each control point. 
We first compute the candidate coordinates for the boundary control points (i.e., where at least one of $i$ and $j$ is in $\{1, 30\}$) using polar coordinates as $\candxycoord{i,j} = (\rho_{i,j} \cos \theta_{i,j}, \rho_{i,j} \sin \theta_{i,j})$, where the radial parameters $\{\rho_{i,j}\}$ are randomly sampled from $[2, 10]$, and the angle parameters $\{\theta_{i,j}\}$ are randomly sampled from $[0, 2 \pi)$ and
sorted in the same order as the control points along the boundary. 
We connect 
these 2D points according 
to the boundary connectivity 
\begin{wrapfigure}{r}{0.49\columnwidth}
	\centering
	\vspace*{-2ex}
	\includegraphics[width=0.49\columnwidth]{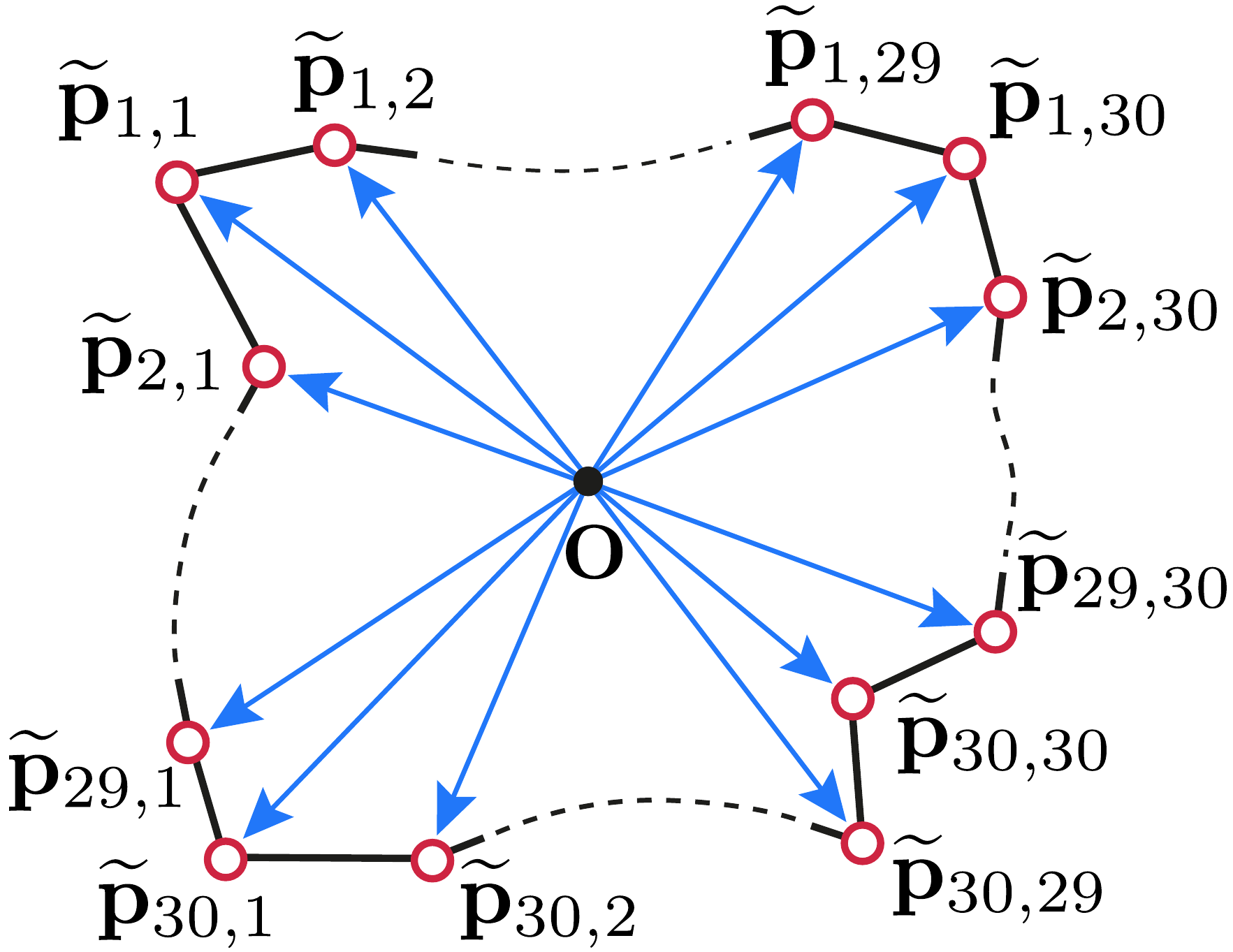}\\
	\vspace*{-1.5ex}
\end{wrapfigure}
of their corresponding control points 
to form a boundary polygon (see inset). To achieve a height field surface, we would like the $(x,y)$-coordinates of the remaining control points to be inside this polygon. To this end, we evenly sample 10000 points $\{\xysample{k}\}$ inside the polygon, and optimize all $\cpxycoord{i,j}$:
\begin{equation*}
\begin{aligned}
 \min\limits_{\{\cpxycoord{i,j}\}}~~ & \sum\nolimits_{(i,j) \in \boundarysamples{}}\|\cpxycoord{i, j} - \candxycoord{i,j}\|_2^2 + \wsampletocp{} \sum\nolimits_{k} \| \xysample{k} - \cpxycoord{\closestcpidx{k}} \|_2^2 \\
& ~~
+ \wcptosample{} \sum\nolimits_{(i,j) \in \interiorsamples{}} \| \cpxycoord{i,j} - \xysample{\closestsampleidx{i,j}}
\|_2^2
+ \wcpxysmooth{} \lcpxysmooth{}. 
\end{aligned}
\end{equation*}
Here $\boundarysamples{}$ and $\interiorsamples{}$ are the index sets of the boundary and non-boundary control points, respectively. The first term above requires the boundary coordinates to be close to their candidate values. The second term penalizes the distance from each sample point $\xysample{k}$ to the $(x,y)$-coordinates of its closest non-boundary control point $\cpoint{\closestcpidx{k}}$. The third term aligns the non-boundary control point coordinates $\cpxycoord{i,j}$ to its closest sample point $\xysample{\closestsampleidx{i,j}}$. 
$\lcpxysmooth{}$ is a fairness term for the coordinates $\{\cpxycoord{i,j}\}$, similar to $\lbsrfair$ in Eq.~\eqref{eq:bsrfair}.
The weights are chosen as $\wsampletocp{} = \wcptosample{} = 1$ and $\wcpxysmooth{} = 10$. We apply a global translation and scaling to all the optimized $(x,y)$-coordinates such that their bounding box is within $[-4, 4]^2$.

Afterwards, we generate the $z$-coordinates of the control points, to achieve a smooth appearance consistent with the style of freeform architecture. Let $\mathbf{z} \in \mathbb{R}^{900}$ be a vector that concatenates all the $z$-coordinates. We construct a matrix $\mathbf{L}$ such that the value $\zfair{\mathbf{z}} = \mathbf{z}^T \mathbf{L}^T \mathbf{L}\mathbf{z}$ is a Laplacian fairness measure of the $z$-coordinates similar to the term $\lbsrfair$ in Eq.~\eqref{eq:bsrfair}. The desirable $z$-coordinates should achieve a small value of $\zfair{\mathbf{z}}$. Therefore, we perform eigendecomposition on the matrix $\mathbf{L}^T \mathbf{L}$ and select the 30 eigenvectors with the smallest eigenvalues. We then linearly combine these vectors using random coefficients from $[-1, 1]$, and scale the resulting vector such that all its components are in $[-5, 5]$. These components are used as the $z$-coordinates.

Finally, we transform the control points with a random freeform deformation (FFD)~\cite{Sederberg:1986:ffd}, which is defined using tricubic B-spline bases with 16$\times$16$\times$16 uniform control points over the domain $[-5, 5]^3$. To induce the deformation, the $x$- and $y$-coordinates of each FFD control point are displaced by a random 2D vector from $[-0.2, 0.2]^2$. Furthermore, to ensure smoothness of the final result, we only retain the surface if the maximum magnitude of its Gaussian curvature is no larger than a  threshold (chosen to be 5 in our experiments). 

The above steps can generate shapes with a wide range of variations, but may not provide good coverage of regular shapes such as ellipsoidal patches. Therefore, we further extract height field patches from random ellipsoids and random cuboids by cutting the shapes with random planes, and fit B-spline surfaces to such patches as part of our dataset.
Fig.~\ref{fig:3d_datasets_examples} shows some examples of B-spline surfaces from our dataset.
More examples are shown in Appendix~\ref{appx:More3dData}.

\begin{figure}[t]
 \includegraphics[width=\linewidth]{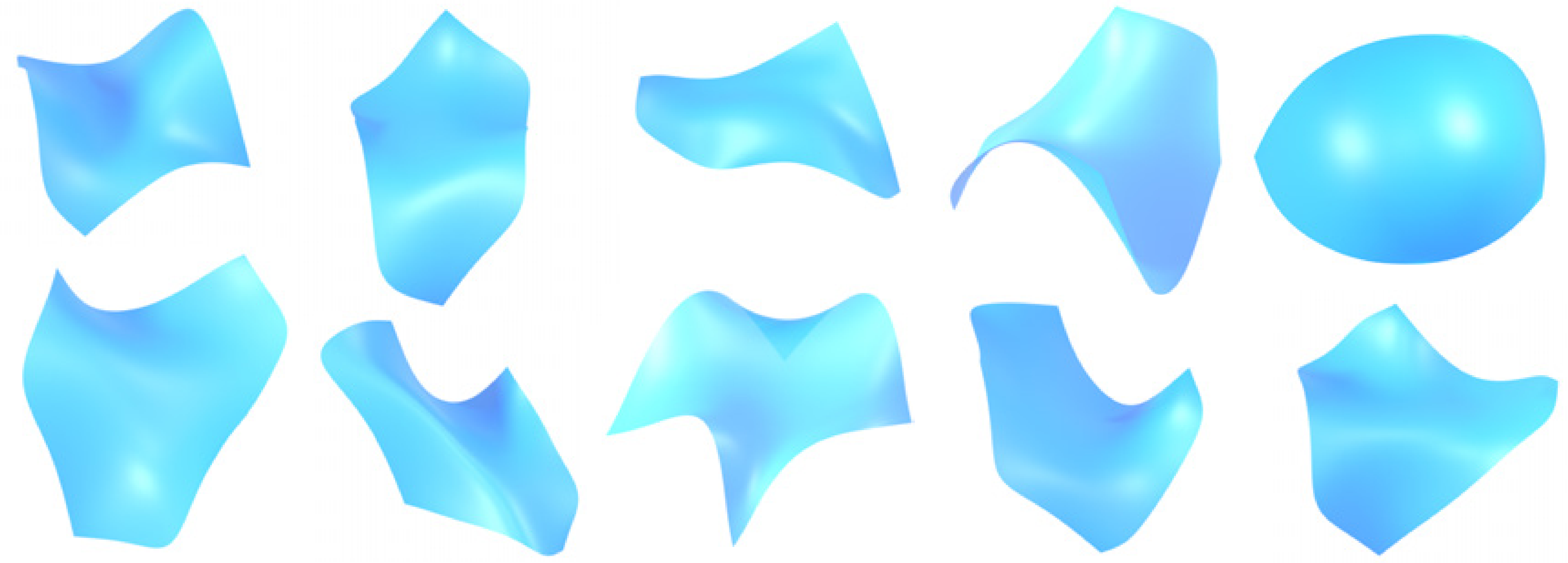}
 \caption{Examples of B-spline surfaces generated using the method described in Sec.~\ref{sec:data_training}.}
 \label{fig:3d_datasets_examples}
 \end{figure}

\myparagraph{Sketch generation} 
To generate an axonometric view of a surface, we first place the camera at $(R/\sqrt{2}, R/\sqrt{2}, R)$ and point it towards the origin, with $R$ being a sufficiently large positive value so that the whole surface is contained in the view volume. Then we rotate the surface with respect to the $z$-axis by an angle $k \cdot \pi/30$ ($k=1, \ldots, 60$) to produce 60 rotated surfaces, and perform orthographic projection on each of them onto the viewing plane to generate the sketched structural lines (i.e., the visible/invisible boundaries and the contour lines) and the visibility/orientation masks. The visibility is determined using the Z-buffer, and the contour lines are determined using the method from~\cite{DeCarlo:2003:contours}. We only retain the sketches that satisfy Assumption~\ref{assump:SurfaceConditions}. For these sketches, we further generate features lines that indicate the PQ mesh layout. To this end, we first specify random conjugate directions at a few random locations on the triangulated surface, and compute a matching dense CDF using the method of~\cite{ComplexRoots:Diamanti:2014}. 
A quad mesh that follows the CDF is then extracted using the method from~\cite{Liu:2011}. Afterwards, we randomly select and trace some mesh edges from the visible region and project them to the viewing plane to obtain the feature lines.
All the structural lines and feature lines are rendered as mentioned in Sec.~\ref{subsec:ui}.
We also randomly select up to six depth samples.  Moreover, to achieve more effective learning, we relabel the B-spline control points such that the corner point $\cpoint{1,1}$ is located in the upper left region of the sketch.
In total, we create 270k data pairs and 54k data pairs for the training and test sets, respectively. Among them, 80k training pairs and 16k test pairs contain occluded regions.

\myparagraph{Network training}
We implement the networks using PyTorch~\cite{NEURIPS2019_Pytorch}.
We first train VD-Net, OD-Net, and BSR-Net in a sequential way.
The VD-Net is first trained for 30 epochs and fixed. Then the OD-Net is trained for 20 epochs and fixed. Afterwards, the BSR-Net is trained for 50 epochs. The CDF-Net is trained separately for 30 epochs. 
Each network is trained on a single Nvidia V100 GPU, using the Adam optimizer~\cite{AdamSolver} with a start leaning rate of $10^{-4}$ and a batch size of 48. The training of VD-Net and OD-Net takes about 16 hours each, and the training of BSR-Net and CDF-Net takes about 10 hours and 8 hours respectively. Finally, VD-Net, OD-Net and BSR-Net are trained together for 20 epochs, with the learning rate for VD-Net and OD-Net reduced to $10^{-5}$.

\section{Experiments and User Studies} \label{sec:results}

\begin{figure}[t]
\centering
 \includegraphics[width=0.85\linewidth]{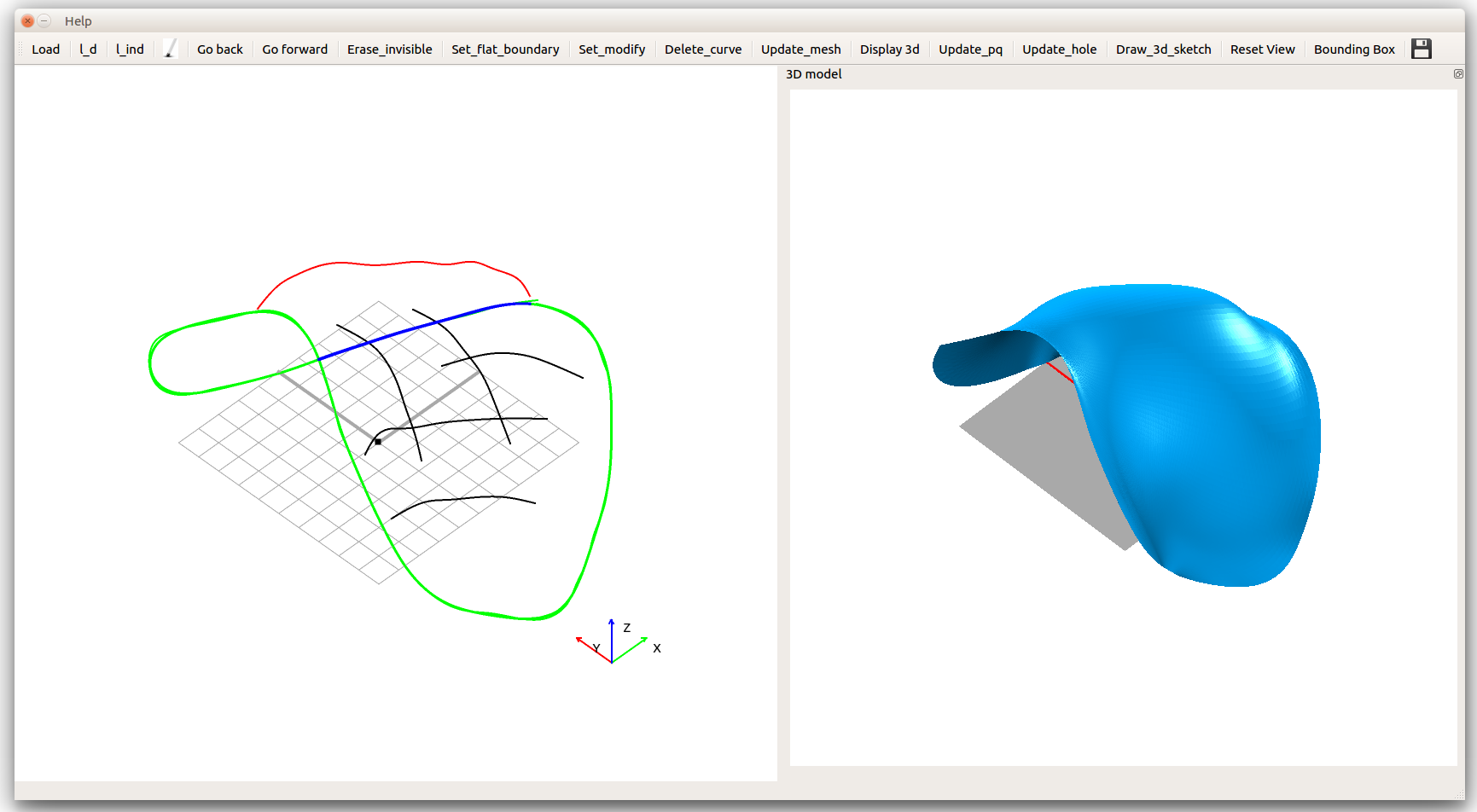}
 \caption{The GUI of our system. The drawing canvas is on the left, and the modeling result is updated in real-time on the right.}
 \label{fig:gui_systerm}
\end{figure}

\begin{figure*}[t]
 \includegraphics[width=\linewidth]{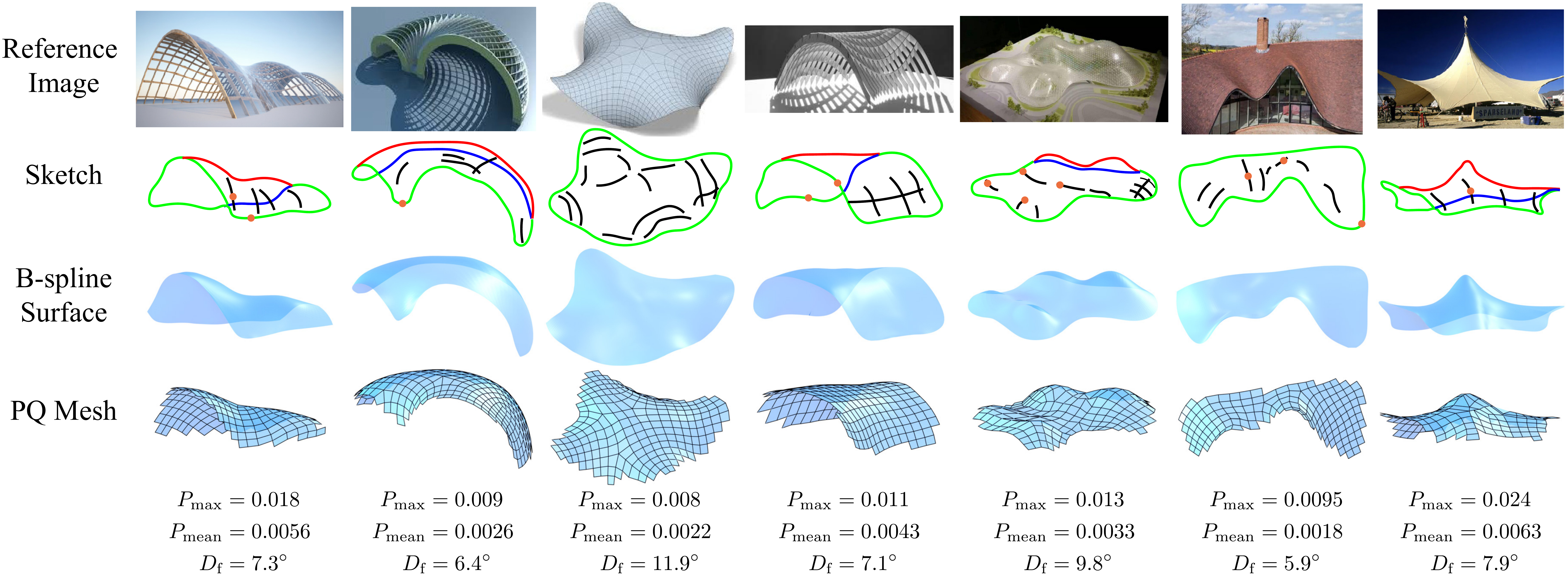}
 \caption{Results on sketches drawn using freeform architecture images as reference.
 For each generated PQ mesh, the max and mean planarity errors $\maxplanarity{},\meanplanarity{}$, and the feature alignment error $\pqedgedist{}$ (Eq.~\eqref{eq:FLalignment}) are also shown.}
 \label{fig:results1}
\end{figure*}

We design a series of experiments and user studies to validate the effectiveness of our system. For intuitive interaction, the UI of our system includes a canvas for sketching and a 3D rendering pane for visualization of the result (see Fig.~\ref{fig:gui_systerm}). In the canvas, the user can select the stroke type and draw the lines using a mouse or a stylus pen. The 3D rendering pane is updated in real-time to show the resulting B-spline surface, and the user can change the view direction to inspect the 3D result. Depth sample points can be specified on either the canvas view or the 3D view, and their depth values can be modified using the mouse scroll wheel.
When the user is satisfied with the shape, they can use a meshing button to generate the PQ mesh.
Examples of user interaction are shown in the accompanying video\footnote{\url{https://www.youtube.com/watch?v=J-YYgaUd1hY}}. 
The system is tested on a PC with an Intel Core i5-8600K CPU at 3.6GHz, an Nvidia GeForce GTX 1060, and 16GB RAM. On average, the inference takes about $40$ms for VD-Net, OD-Net and BSR-Net, and about $28$ms for CDF-Net. The B-spline surface refinement takes about $100$ ms. The PQ mesh extraction takes about $3.8$s; since it only needs to be executed occasionally, this does not affect the interactivity of the system.

\subsection{Evaluation}
We evaluate the quality of the results generated by our system from two aspects. First, we measure the difference between the recovered B-spline surface and the ground-truth surface using their Chamfer distance based on point positions and normals, denoted as $\chamferdist{}$ and $\normaldist{}$ respectively:
\begin{align}
    \chamferdist{} & = 
    \frac{\sum_{\recsurfsample{i} \in \recsurfsampleset{}} \| \recsurfsample{i} - \gtsurfsample{\gtsurfcpidx{i}} \|^2
    + \sum_{\gtsurfsample{j} \in \gtsurfsampleset{}} \| \gtsurfsample{j} - \recsurfsample{\recsurfcpidx{j}} \|^2}{|\recsurfsampleset{}| + |\gtsurfsampleset{}|},
    \label{eq:ChamferDist}\\
    \normaldist{} & = 
    \frac{1}{|\recsurfsampleset{}| + |\gtsurfsampleset{}|}[\sum\nolimits_{\recsurfsample{i} \in \recsurfsampleset{}} \anglefunc(N(\recsurfsample{i}), N(\gtsurfsample{\gtsurfcpidx{i}}))
    \nonumber\\
    &\qquad\qquad\qquad~~ + \sum\nolimits_{\gtsurfsample{j} \in \gtsurfsampleset{}} \anglefunc(N(\gtsurfsample{j}), N(\recsurfsample{\recsurfcpidx{j}}))].
    \label{eq:NormalDist}
\end{align}
Here $\recsurfsampleset{}, \gtsurfsampleset{}$ are the sample point sets for the recovered surface and the ground-truth surface respectively,  $\gtsurfsample{\gtsurfcpidx{i}}$ is the closest point to $\recsurfsample{i}$ on the ground-truth surface, $\recsurfsample{\recsurfcpidx{j}}$ is the closest point to $\gtsurfsample{j}$ on the recovered surface, $N(\cdot)$ denotes the normal vector at a point, and $\anglefunc$ is the angle between two vectors: 
\begin{equation}
\anglefunc(\vecsymbol_1, \vecsymbol_2)
= \arccos \left( \left|\vecsymbol_1 \cdot \vecsymbol_2\right|/({\|\vecsymbol_1\|\cdot\|\vecsymbol_2\|}) \right).
\label{eq:AngleFunc}
\end{equation}
As there is no public sketch dataset for PQ meshes, the evaluation is performed using the test dataset created in Sec.~\ref{sec:data_training}.
Table~\ref{tab:evaluation} shows the average values of $\chamferdist{}$ and $\normaldist{}$ using our method on the test dataset. 
Since the sketches are represented as RGB images, for comparison we also include results from alternative learning-based methods~\cite{atlasnet2018,Deng:2020:3DV,smirnov2021patches} that can generate a disk-topology surface from a single sketch image.
Among them, AtlasNet~\cite{atlasnet2018} represents the generated shape as a collection of parametric surface patches, and uses a loss function that penalizes the chamfer distance between sample point sets from the generated surface and the ground-truth surface. 
We test it using  one, two and five patches, respectively. We also test a variant that uses a single patch and with additional loss terms to align the boundaries of the generated and ground-truth surfaces (see Appendix~\ref{appx:SurfPredictionComparison}).
\cite{Deng:2020:3DV} uses a similar representation as AtlasNet, with additional loss terms to enforce consistency between the patches.
We test it using one and four patches, respectively. 
\cite{smirnov2021patches} represents the shape with a template consisting of Coons patches.
We test it using 10$\times$10 bi-cubic Coons patches.
All three methods are re-trained using our training dataset (see Appendix~\ref{exp:training_details}). 
Table~\ref{tab:evaluation} shows their average values of $\chamferdist{}$ and $\normaldist{}$ on our test dataset.
Appendix~\ref{appx:SurfPredictionComparison} further shows two example sketches from the test dataset and the resulting surfaces using each method.
These results show that our method is effective in recovering the underlying B-spline surface shape, and produces more accurate results than the alternative approaches. 

In addition, we evaluate the quality of the resulting PQ mesh with two metrics. 
We measure the planarity error of each face using the distance between its two diagonal lines divided by the average edge length of the mesh, and calculate the maximum planarity error $\maxplanarity{}$ and the mean planarity error $\meanplanarity{}$ across the whole mesh.
We also measure the alignment between the PQ mesh edge directions and the feature lines from the input sketch. 
To do so, we first project the visible PQ mesh edges onto the viewing plane. Then we densely sample a set of points $\flsampleset{} = \{\flsample{i}\}$ from the feature lines, and compute an angle-based alignment error:
\begin{equation}
    \pqedgedist{} = \sum\nolimits_{\flsample{i} \in \flsampleset{}} \anglefunc(T(\flsample{i}), H(\edgecp{i})), 
    \label{eq:FLalignment}
\end{equation}
where $\edgecp{i}$ is the closest point on the projected edges from the sample point $\flsample{i}$, $T(\cdot)$ and $H(\cdot)$ denote the feature line tangent direction and the projected edge direction respectively, and $\anglefunc{}$ is the angle function defined in Eq.~\eqref{eq:AngleFunc}.
On the test dataset, our method achieves the the average error metrics of 
$\avgmaxplanarity{} = 0.0169$, $\avgmeanplanarity{} = 0.0042$ and 
$\avgpqedgedist{} = 7.8^{\circ}$. It shows that our method can produce PQ meshes that align well with the input feature lines and with low planarity errors.

\begin{figure*}[t]
 \includegraphics[width=\linewidth]{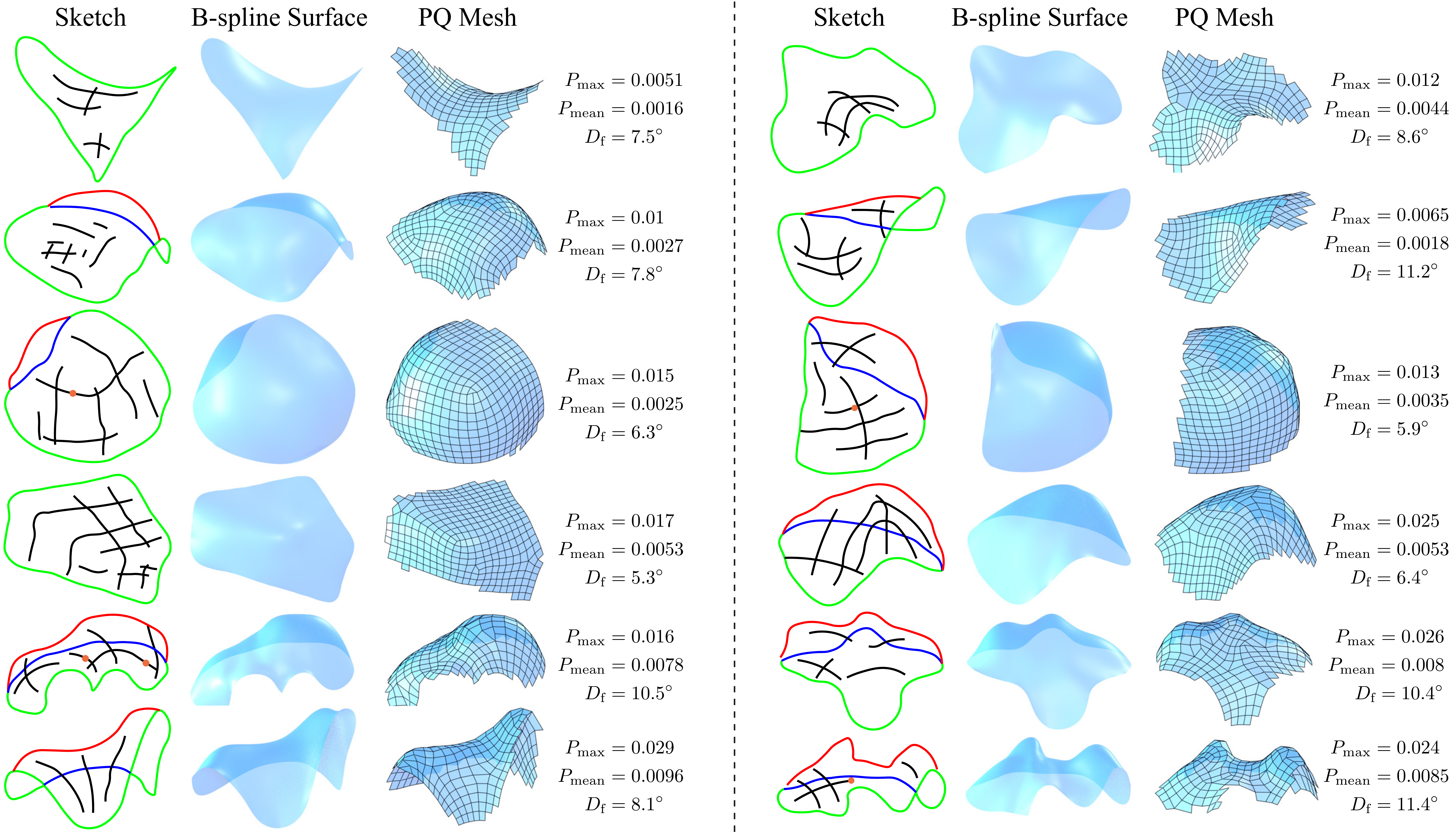}
 \caption{Results on sketches drawn without reference images.
 For each generated PQ mesh, the max and mean planarity errors $\maxplanarity{},\meanplanarity{}$, and the feature alignment error $\pqedgedist{}$ (Eq.~\eqref{eq:FLalignment}) are also shown.}
 \label{fig:results2}
\end{figure*}

\renewcommand\arraystretch{1.4}
\begin{table}[t]
 \caption{Average values of $\chamferdist{}$ and $\normaldist{}$ (Eqs.~\eqref{eq:ChamferDist} \& \eqref{eq:NormalDist}) on our test dataset using different methods to predict the underlying surface. For~\cite{atlasnet2018} and \cite{Deng:2020:3DV}, the numbers in the sub-column headings indicate the number of patches.}
 \centering
\setlength{\tabcolsep}{2pt}
\begin{tabular}{c|c|c|c|c|c|c|c|c}
\hline
\multirow{2}{*}{}&
\multicolumn{4}{c|}{AtlasNet~\cite{atlasnet2018}}&
\multicolumn{2}{c|}{\multirow{1}{*}{\cite{Deng:2020:3DV}}} &
\multirow{2}{*}{\cite{smirnov2021patches}} &
\multirow{2}{*}{Ours}\\
\cline{2-7}
   & 1&2 & 5 & variant & 1& 4&\multicolumn{1}{c|}{}& \\
\hline
 $\chamferdist{}$ ($\times 10^{-2}$)& 4.69& 4.01 & 3.42 &3.83  & 6.38 & 9.45 & 7.53 & \textbf{0.98}\\
\hline
 $\normaldist{}$ & ~27.5$^{\circ}$& ~22.5$^{\circ}$ & ~21.6$^{\circ}$ & ~18.7$^{\circ}$ & ~18.6$^\circ$& ~21.5$^\circ$ & ~15.8$^{\circ}$ & ~\textbf{7.23}$^{\boldsymbol{\circ}}$\\
\hline
\end{tabular}
 \label{tab:evaluation}
\end{table}

Figs.~\ref{fig:results1} and  \ref{fig:results2} show some examples of sketches outside our dataset, and their corresponding B-spline surfaces and PQ meshes generated with our system.
The sketches in Fig.~\ref{fig:results1} are drawn according to rendered images of architectural shapes or real architectural photos, while the ones in Fig.~\ref{fig:results2} are drawn without reference shapes. 
In both figures, the resulting shapes align well with the structural lines and achieve the correct occlusion relation conveyed by the sketch.
The figures also include the error metrics $\maxplanarity$, $\meanplanarity{}$ and $\pqedgedist{}$ for the PQ meshes, which show that the meshes have low planarity errors and align well with the feature lines.

\begin{figure}[t]
 \includegraphics[width=\linewidth]{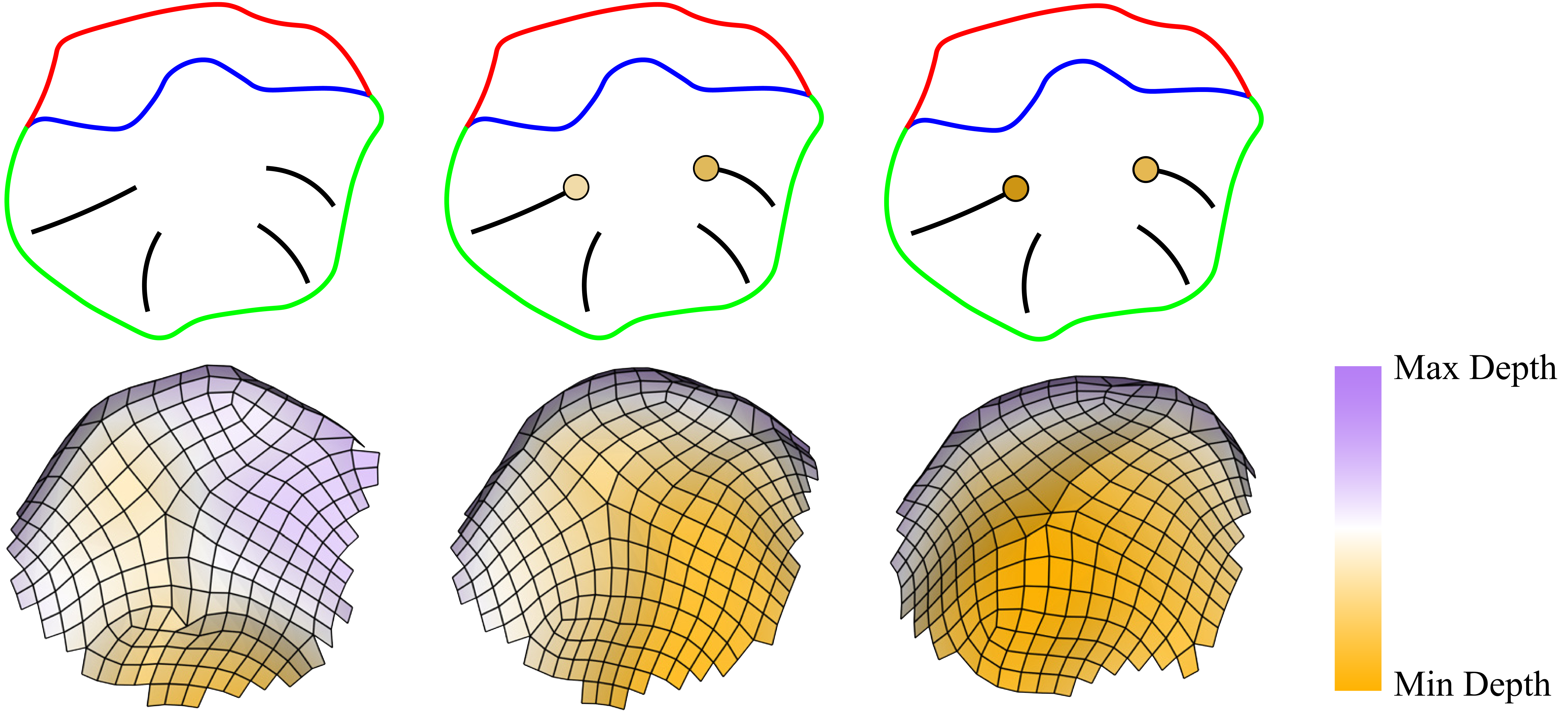}
 \caption{Three sketches with the same strokes but different numbers of depth samples, and the resulting PQ meshes. Both the depth sample values and the PQ mesh depth values are visualized using color coding.}
 \label{fig:DepthSampleControl}
\end{figure}

In our system, the feature lines and depth samples enable the user to fine-tune the surface shape and the PQ mesh layout. Figs.~\ref{fig:DepthSampleControl} and \ref{fig:FeatureLineControl} provide examples of such controls. Fig.~\ref{fig:DepthSampleControl} shows three inputs with the same sketch but different depth samples: the first one has no depth sample, while the other two have depth samples at the same locations but with different depth values. We can see that the depth values of the resulting PQ meshes are consistent with the specified values at the depth sample locations, which shows the effectiveness of depth samples for shape control. 
Fig.~\ref{fig:field_sketch_field} shows a series of sketches with the same structural lines and incrementally more feature lines. We can see that each additional feature line influences the PQ mesh shape and the mesh edge layout in its surrounding area, providing an intuitive tool to control the PQ mesh.

\subsection{Ablation Study}

\myparagraph{Depth prediction} Our system predicts the depth and normal maps using two separate network modules for the visible regions and the occluded regions, respectively. We test an alternative network that predicts all the depth and normal maps jointly with a U-Net structure: the input sketch, depth samples and all the masks are fed to an encoder with the same architecture as the ones used in VD-Net and OD-Net; there are two decoder branches that generate the  visible and occluded depth maps separately, each with the same structure as the decoder in VD-Net and OD-Net, respectively.
In addition, to validate the necessity of the orientation information in the masks, we also test an alternative approach that replaces the four input visibility-orientation masks to VD-Net and OD-Net with two masks that indicate the visible and occluded regions but not their orientations.
We compare the depth maps generated by our method and the alternative approaches on the test dataset. For each method, we measure the accuracy of the predicted visible depth values $\{d_i^{\text{v}} \mid i \in \vispixset{}\}$ and occluded depth values $\{d_j^{\text{o}} \mid j \in  \occpixset{}\}$ using their mean squared errors compared to the respective ground truth values $\{d_i^{\text{v,gt}}\}$ and $\{d_i^{\text{o,gt}}\}$:
\begin{figure}[t]
	\includegraphics[width=\linewidth]{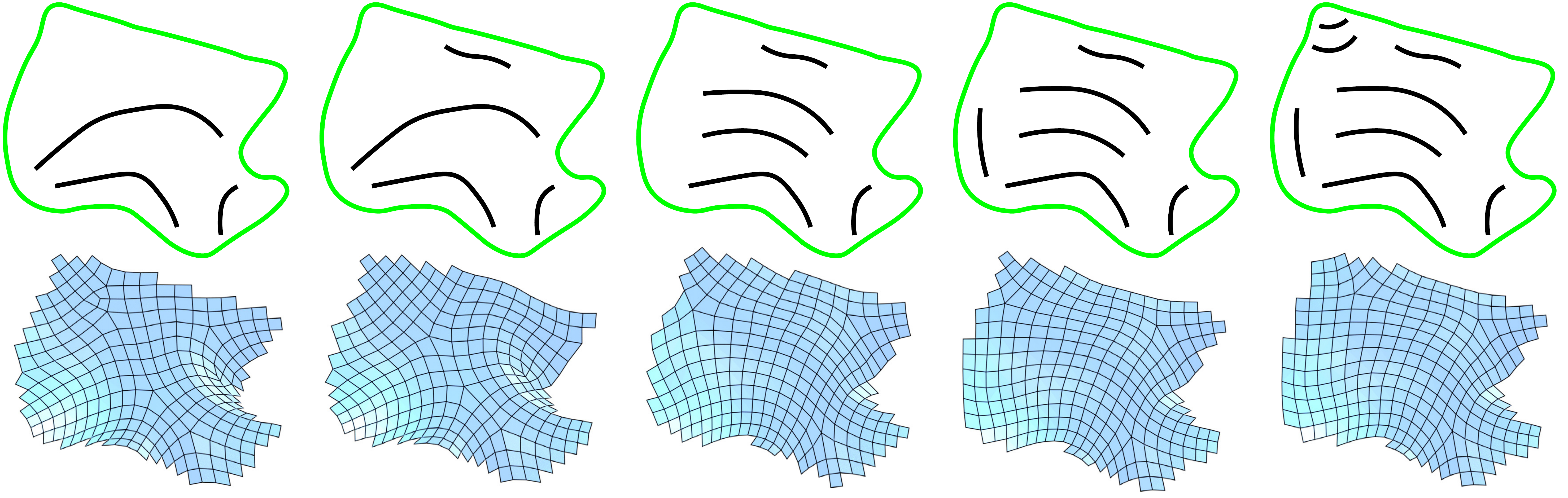}
	\caption{An example where the resulting PQ mesh layout is fine-tuned by incrementally adding feature lines.}
	\label{fig:FeatureLineControl}
\end{figure}
\renewcommand\arraystretch{1.4}
\begin{table}[t!]
 \caption{Average values of depth error metrics $\visdepthmse{}, \occdepthmse{}$ (Eq.~\eqref{eq:DepthMSE}) on the test dataset using our method and two alternative methods: joint inference of visible and occluded depth maps (``Joint''), and using visible and occluded masks without orientation information (``Two-Masks'').}
 \centering
 \begin{tabular}{cccc}
 \hline
  & Ours & Joint & Two-Masks \\
 \hline
 $\visdepthmse{}$ ($\times 10^{-2}$) &  \textbf{0.6} & 0.84 & 1.6\\
 \hline
 $\occdepthmse{}$ ($\times 10^{-2}$) & \textbf{0.12} & 0.64 & 1.2\\
 \hline
 \end{tabular}
 \label{tab:multistep}
\end{table}
\begin{table}[t!]
 \caption{
Average values of $\chamferdist{}, \normaldist{}$ (Eqs.~\eqref{eq:ChamferDist} \& \eqref{eq:NormalDist}) on the test dataset using our B-spline surface prediction method and two alternative methods: direct prediction of control points using a ResNet, and using the BSR-Net alone without further optimization.}
 \centering
 \begin{tabular}{cccc}
 \hline
 & Ours & ResNet& BSR-Net Only \\
 \hline
 $\chamferdist{}$ ($\times 10^{-2}$) & \textbf{0.98} & 12.7 & 1.98 \\
 \hline
 $\normaldist{}$ ($\times 10^{-2}$) & ~~\textbf{7.23}$^{\boldsymbol{\circ}}$ & ~~20.2$^{\circ}$ &  ~~12.3$^{\circ}$ \\
 \hline
 \end{tabular}
 \label{tab:AblationBSR}
\end{table}
\begin{table}[t!]
 \caption{Average values of CDF smoothness ($\cdfsmoenergy{}$), max/mean PQ mesh planarity errors ($\maxplanarity{}$/$\meanplanarity{}$), and feature line alignment error ($\pqedgedist{}$) on the test set, with and without the term $\lcdfcons$ in the CDF-Net loss function.}
 \centering
 \begin{tabular}{ccccc}
 \hline
 & $\cdfsmoenergy{}$ & $\maxplanarity{}$ & $\meanplanarity{}$ &  $\pqedgedist{}$ \\
 \hline
 With $\lcdfcons$ (Ours)& \textbf{637} & \textbf{0.0169} & \textbf{0.0042}  &  \textbf{7.8}$^{\boldsymbol{\circ}}$\\
 \hline
 Without $\lcdfcons$& 775& 0.0174 & 0.0043 &  8.9$^{\circ}$\\
 \hline
 \end{tabular}
 \label{tab:consistent_field}
\end{table}
\begin{equation}
\visdepthmse{} = \frac{\sum_{i\in \vispixset{}} (d_i^{\text{v}} - d_i^{\text{v,gt}})^2}{| \vispixset{} |},~~ 
\occdepthmse{} = \frac{\sum_{j\in \occpixset{}} (d_j^{\text{o}} - d_j^{\text{o,gt}})^2}{ | \occpixset{} | }.
\label{eq:DepthMSE}
\end{equation}
Table~\ref{tab:multistep} shows the average values of $\visdepthmse{}$ and $\occdepthmse{}$ on the test dataset for each method. It verifies that our method generates more accurate depth maps than the alternative approaches.
Our approach first handles the visible depth map, which is typically an easier problem than the occluded depth map due to its larger area as well as the feature lines that provide geometric information. Afterwards, the predicted visible depth map provides more cues for predicting the occluded depth map. We note that similar approaches have been adopted in existing 3D reconstruction methods such as~\cite{Wu:marrnet2017,Wu:2018:shaped}.
In addition, the orientation information is an important cue for the surface shape and its relation with the viewing plane. Hence the inclusion of orientation in the input masks for OD-Net and VD-Net improves the prediction accuracy.

\firstparagraph{B-spline surface prediction} 
Our system first uses BSR-Net to infer an initial B-spline surface and further optimizes it  to align with the predicted depth maps. To validate the necessity of this setup, we test two alternative approaches: (1)~using BSR-Net to infer the B-spline control points but without the subsequent optimization; (2)~using a ResNet-18~\cite{He2015} to directly infer the control points from the sketch without using the predicted depth maps. 
Table~\ref{tab:AblationBSR} shows the average values of the B-spline surface error metrics $\chamferdist{}, \normaldist{}$ (see Eqs.~\eqref{eq:ChamferDist} and \eqref{eq:NormalDist}) on the test dataset using each method. 
Both alternative approaches lead to less accurate results than our method. 
Fig.~\ref{fig:BSplineComparison} shows some sketches from the test dataset and their resulting surfaces using each approach. We can see that direct inference with ResNet may produce surfaces that deviate from the structural lines. The BSR-Net, with the help of the predicted depth maps, can better align the surfaces with the structural lines than ResNet, although there can still be deviations in some local regions. The optimization after BSR-Net further reduces the local deviation and produces a result that aligns well with the structural lines.

\begin{figure}[t]
 \includegraphics[width=\linewidth]{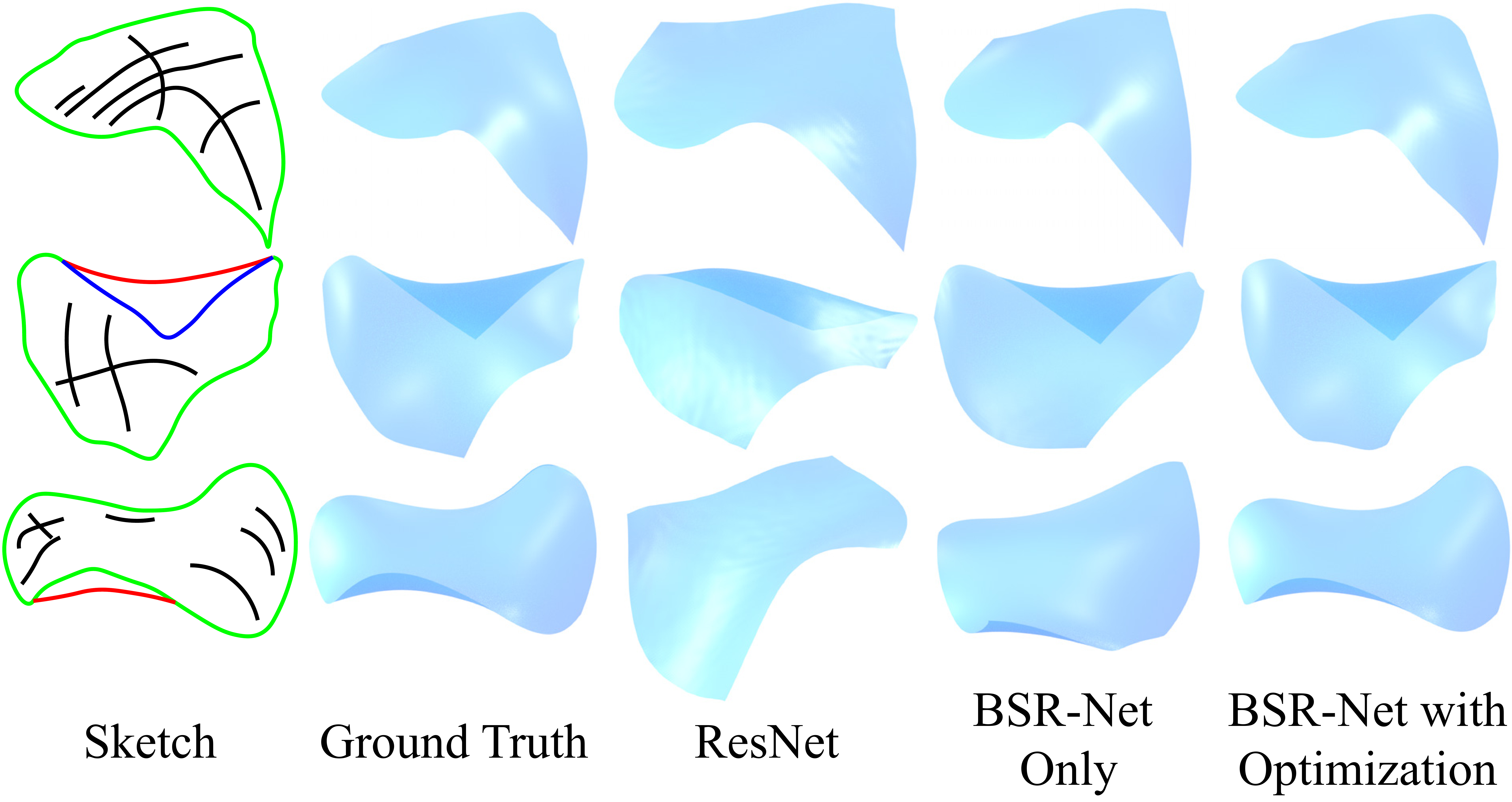}
 \caption{Comparison of the predicted B-spline surfaces using a ResNet, the BSR-Net alone, and the BSR-Net with further optimization. BSR-Net with optimization produces more accurate results.}
 \label{fig:BSplineComparison}
\end{figure}

\begin{figure}[t]
\includegraphics[width=\linewidth]{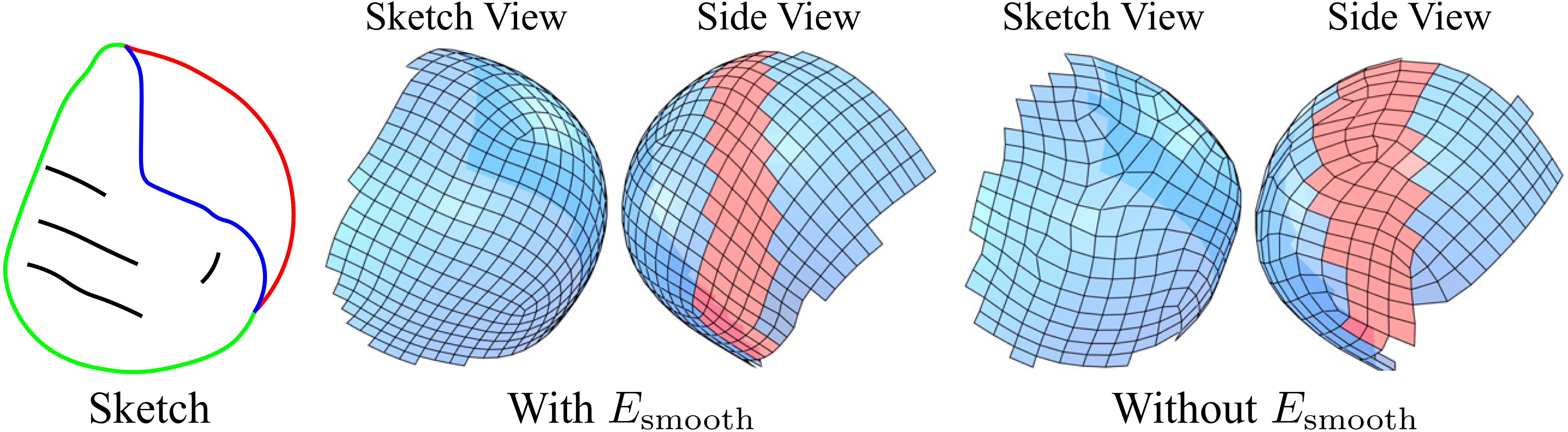}
\caption{PQ mesh results with and without the term $\lcdfcons$ in the CDF-Net loss function. In this example, the lack of $\lcdfcons$ leads to a more irregular layout around the contour (shown in red in the side views).}
\label{fig:field_sketch_field}
\end{figure}

\myparagraph{CDF prediction} 
Our loss function for the CDF-Net includes a term $\lcdfcons$ that enforces consistency of the CDF across the contour line. 
We test an alternative approach that does not include this term in the loss function. 
For each method, we compute a smoothness measure $\cdfsmoenergy{}$ of the CDF on the predicted surface using the Dirichlet energy of its PolyVector complex polynomial coefficients as explained in~\cite{Diamanti2015Integrable}. Table~\ref{tab:consistent_field} shows the average value of $\cdfsmoenergy{}$ using each method on the test dataset. We can see that discarding the term $\lcdfcons$ results in a CDF that is less smooth on average. 
Fig.~\ref{fig:field_sketch_field} further shows an example of PQ meshes from each method. It can be seen that the removal of $\lcdfcons$ leads to a less smooth PQ mesh layout especially around the contour region.
Table~\ref{tab:consistent_field} also shows the average values of the max and mean planarity errors $\maxplanarity{}, \meanplanarity{}$ and the alignment error $\pqedgedist{}$ (see Eq.~\eqref{eq:FLalignment}) for the resulting PQ meshes using each method. It shows that the lack of smoothness for the CDF also affects the quality of the PQ meshes and increases the planarity errors and the deviation from the sketched feature lines.

\subsection{User Studies} \label{subsec:userstudy}
We conducted user studies to evaluate the usability of our system. 
The participants were 12 senior undergraduate or graduate architecture students with experience in 3D 
modeling software. 
They were given a short tutorial for the system and practiced with a few examples. The training  
\begin{wrapfigure}{r}{0.5\columnwidth}
	\centering
	\vspace*{-1ex}
	\includegraphics[width=0.5\columnwidth]{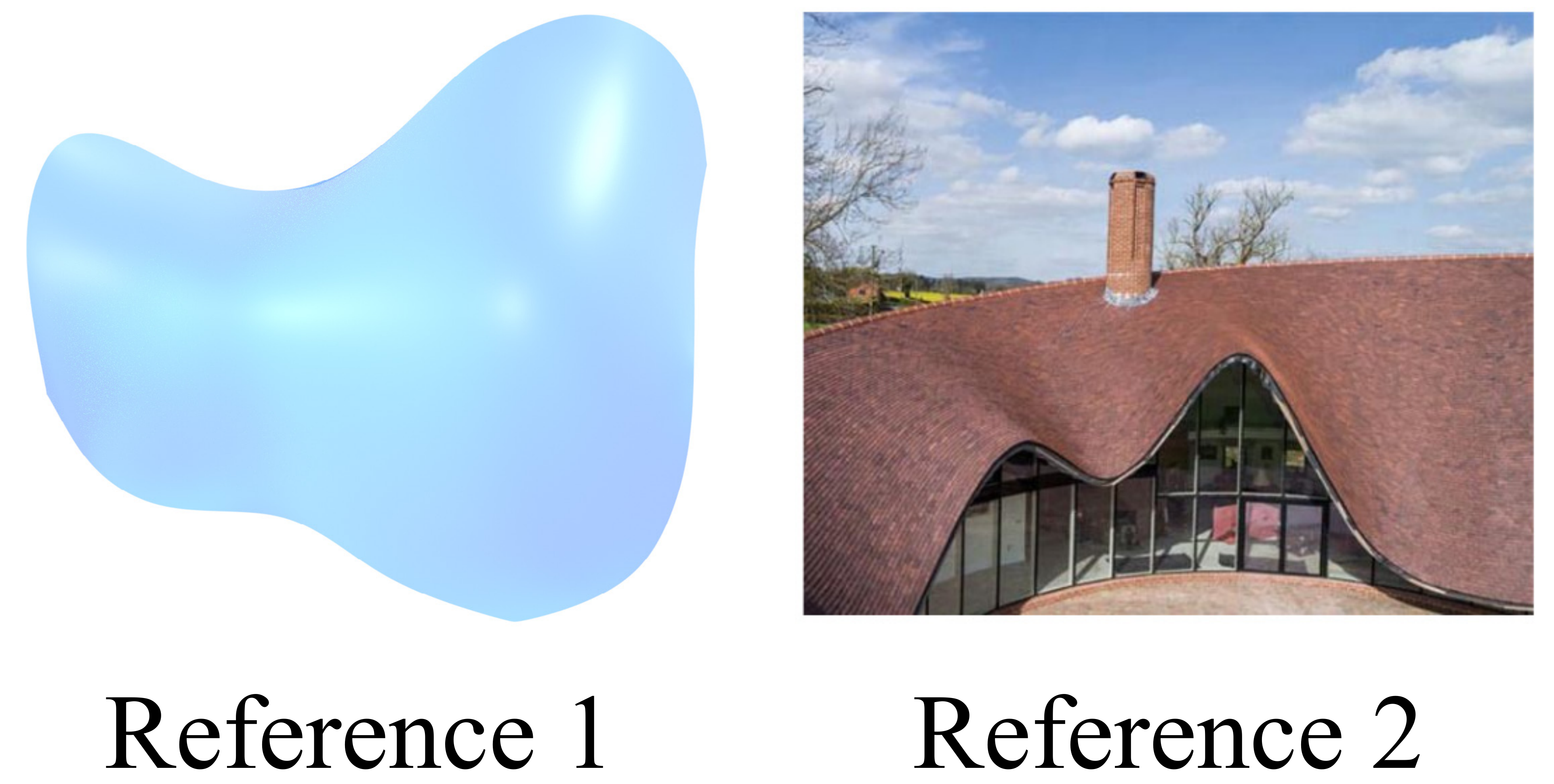}
	\vspace*{-3ex}
\end{wrapfigure}
took about 20 minutes for each participant. Afterwards, the participants were asked to model three PQ meshes. 
The first two need to be modeled according to a given rendered image of an architectural surface and a real architectural photo, respectively (see inset). 
The third one was required to have self-occlusion and the participants could design any shape they liked.
On average, each participant spent 10, 15 and 15 minutes modeling the three shapes, respectively. The modeling time was spent mostly on fine-tuning the shape until the user was satisfied with the result.  Fig.~\ref{fig:UserStudyExamples} shows the designs from two participants for these tasks. We can see that the users were able to create designs similar to the reference shapes. More results are shown in Appendix~\ref{appx:UserStudyMoreExamples}.

\begin{table}[t]
\centering
\caption{Average scores for user study questions. The individual scores can be found in Appendix~\ref{appx:answers}.}
\begin{tabular}{cccccc}
\hline
Question & Q1 & Q2 & Q3 & Q4 & Q5 \\
\hline
Average Score &3.84 & 3.92 & 4 & 4.07 & 4.23 \\
\hline
\end{tabular}
\label{tab:feedback_from_user}
\end{table}

In addition, we designed the following questions collect to feedback from the participants:
\begin{enumerate}
 \item Is the system intuitive compared to other 3D modeling tools you have used, such as Blender, Sketch-up and Rhino3D?
 \item Does the generated shape meet your expectations?
 \item Are the stroke lines effective in helping you create a model?
 \item Does the feature of drawing occluded regions help your modeling?
 \item Do the feature lines help you control the PQ mesh layout?
\end{enumerate}
For each question, each participant was required to choose an answer from five options corresponding to a score of 1 to 5, where a higher score indicates better agreement. The detailed options and the participants' answers can be found in Appendix~\ref{appx:answers}. Table~\ref{tab:feedback_from_user} shows the average score for each question. Overall, the system received positive feedback. 

The participants were also asked to provide free-text comments and suggestions for improvement.
Many participants commented that the capability to model the surface and PQ mesh simultaneously from a sketch is an attractive feature compared with other 3D modeling software.
The main suggestions include: (1)~development of plugins to integrate our system into existing modeling software; (2)~additional constraints on the structural properties (i.e., designing PQ meshes that are also structurally sound); (3)~capability to design multiple surfaces in a single sketch.

 \begin{figure}[t]
 \centering
 \includegraphics[width=\columnwidth]{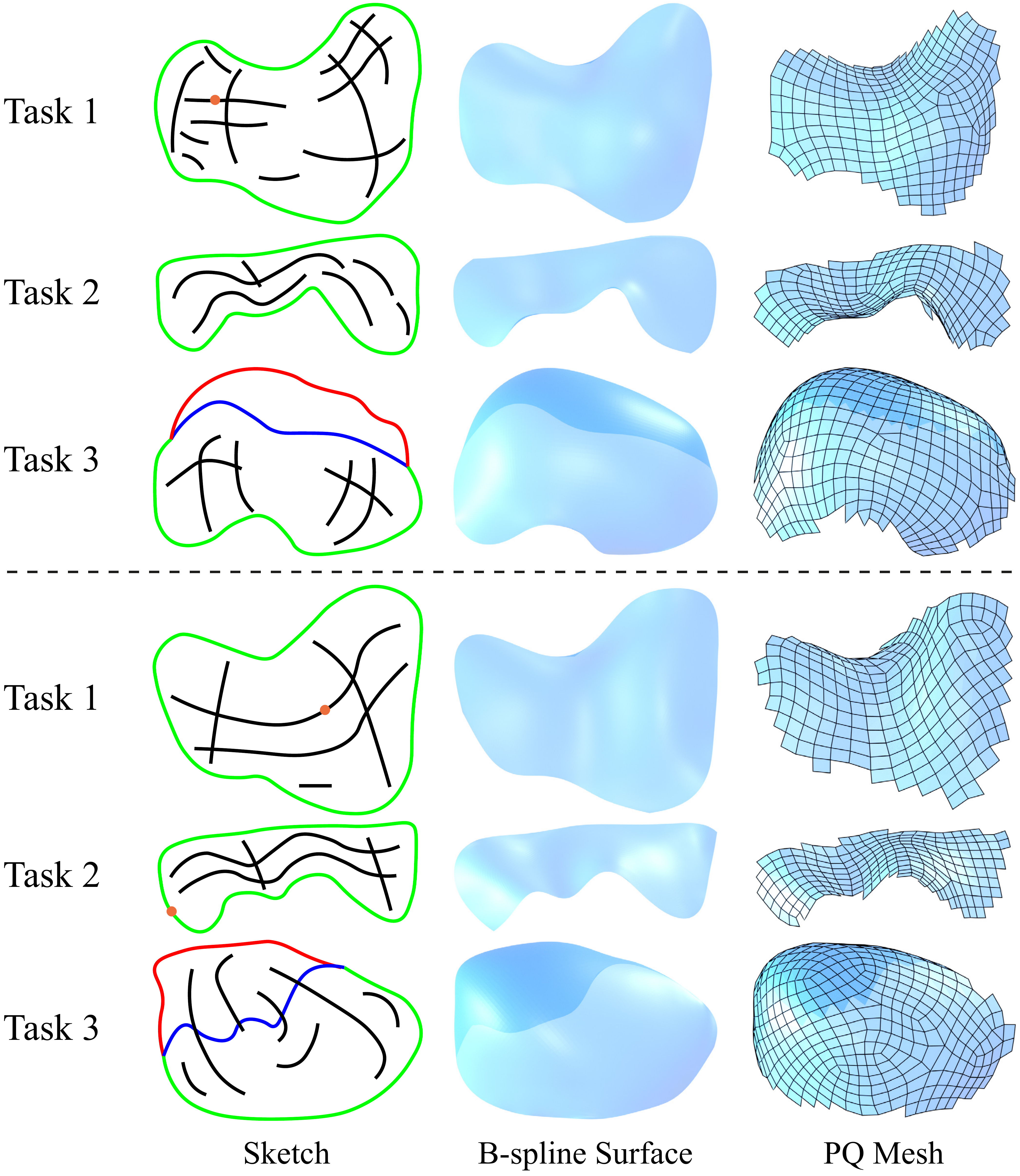}
 \caption{The designs from two participants of the user study for the surface modeling tasks.}
 \label{fig:UserStudyExamples}
 \end{figure}

\section{Conclusion and Discussion}
In this work, we present a sketch-based modeling system for freeform roof-like structures represented as PQ meshes. 
Our system allows the user to draw structural lines for the surface boundary and contours with annotation for occlusion, as well as sparse feature lines that indicate the directions of PQ mesh edges. Utilizing a neural network trained with a synthetic dataset, the system can robustly generate a PQ mesh from a single sketch. The effectiveness and usability of the system are verified with extensive experiments and user studies.
Our learning-based system provides notable benefits in efficiency and can generate the underlying surface and its CDF in real-time; although some existing non-learning-based approaches may be adapted to achieve similar results, this would involve a large-scale non-linear optimization problem for \emph{every} input sketch to find a surface and a CDF consistent with the sketch, which can be substantially slower.
Our work shows that sketch-based modeling can be an effective tool for freeform architectural design where the components need to satisfy additional geometric constraints. 
It enables a novel design paradigm that complements the existing workflow, allowing the highly constrained 3D panel layout to be directly generated using its appearance from a certain view angle, without the need for an existing reference surface.

We envisage several avenues for future research that further improves the system. 
First, the system can be extended to handle more complex surface shapes and configurations such as non-disk topology and multiple layers of occlusions. These can potentially be achieved by adopting topology-agnostic surface representations such as signed distance fields~\cite{Park_2019_CVPR} and adapting the network modules accordingly.
Secondly, our current data generation procedure is based solely on the requirement of a smooth appearance. It can be further developed to incorporate other criteria such as structural properties and design preferences.
Thirdly, due to the PQ meshing algorithm we use, the sketched boundary may not completely align with the boundary edges of the resulting PQ mesh. This can potentially be addressed by introducing optional constraints of boundary edge alignment in a future work.
Fourthly, more fine-grained shape controls such as specification of singular vertices or sharp edges will be helpful additions to the system, which can potentially be achieved by augmenting the training data to include such elements.
Finally, our approaches can be extended to other types of structures and panel layouts in freeform architectural design, such as polyhedral surfaces and developable panels~\cite{pottmann2015architectural}.

\section*{Acknowledgments}
The authors thank Yue Peng for helping to conduct the experiments.
This work was supported by National Natural Science Foundation of China (No. 62122071), the Youth Innovation Promotion Association CAS (No. 2018495), and ``the Fundamental Research Funds for the Central Universities'' (No. WK3470000021).

%\appendices
\bibliographystyle{IEEEtran}
\bibliography{src/sketchCNNBib}
%\vspace*{-5mm}
%\newpage
\begin{IEEEbiography}[{\includegraphics[width=1in]{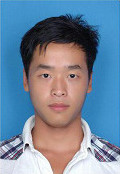}}]{Zhi Deng} is a PhD student at the School of
   	Data Sciences, University of Science and
   	Technology of China. His
   	research interests include Computer Vision and
   	Computer Graphics.
\end{IEEEbiography}

\vspace*{-2em}

\begin{IEEEbiography}[{\includegraphics[width=1in]{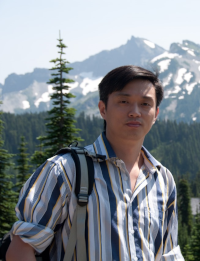}}]{Yang Liu}
is a principal researcher at Microsoft Research Asia. He received his Ph.D. degree from The University of Hong Kong, Master and
Bachelor degrees from University of Science and Technology of China. His recent research focuses on geometric computation and learning-
based geometry processing and generation. He is an associate editor of IEEE Transactions on Visualization and Computer Graphics.
\end{IEEEbiography}

\vspace*{-2em}

\begin{IEEEbiography}[{\includegraphics[width=1in]{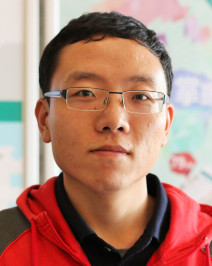}}]{Hao Pan} is a senior researcher at Internet Graphics Group, Microsoft Research Asia.
He received the BEng degree in software engineering from the Shandong University, and the PhD degree in computer science from The
University of Hong Kong. His research interests include computer graphics, 3d computer vision, and geometry modeling and processing.
\end{IEEEbiography}

\vspace*{-2em}

\begin{IEEEbiography}[{\includegraphics[width=1in]{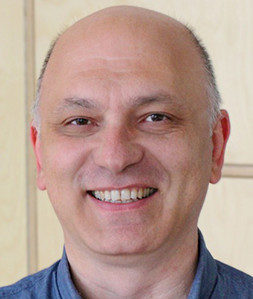}}]{Wassim Jabi}
is a professor at the Welsh School of Architecture, Cardiff University where he leads the digital design area. He earned his B.Arch. from the American University of Beirut and his M.Arch. and Ph.D. from the University of Michigan. His current research is at the intersection of parametric design, the representation of space, building performance simulation, and robotic fabrication in architecture. He is a member of the editorial board of the International Journal of Architectural Computing (IJAC).
\end{IEEEbiography}

\vspace*{-2em}

\begin{IEEEbiography}[{\includegraphics[width=1in]{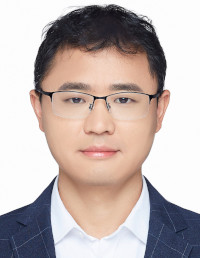}}]{Juyong Zhang}
is a professor in the School of Mathematical Sciences at University of Science and Technology of China. He received the BS degree from the University of Science and Technology of China in 2006, and the PhD degree from Nanyang Technological University, Singapore. His research interests include computer graphics, computer vision, and numerical optimization. He is an associate editor of IEEE Transactions on Multimedia and The Visual Computer.
\end{IEEEbiography}

\vspace*{-2em}    
    
\begin{IEEEbiography}[{\includegraphics[width=1in]{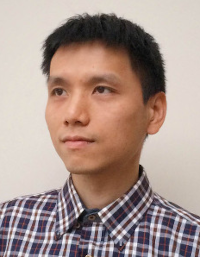}}]{Bailin Deng}
is a senior lecturer in the School of Computer Science and Informatics at Cardiff University. He received the BEng degree in computer software (2005) and the MSc degree in computer science (2008) from Tsinghua University (China), and the PhD degree in technical mathematics (2011) from Vienna University of Technology (Austria). His research interests include geometry processing, numerical optimization, computational design, and digital fabrication. He is a member of the IEEE.
\end{IEEEbiography}

\vfill

\clearpage
\appendices

\section{Proof for Proposition~\ref{prop:Orientation}} 
\label{appx:Proof}
We first prove the following lemmas.

\begin{lemma}
\label{lemmar:1}
If two minimal regions share a boundary line, then one of them must correspond to an occluded part of the surface.
\end{lemma}
\begin{proof}
Assume that neither of the minimal regions corresponds to an occluded part. Then each of them corresponds to a separate visible region, and the shared boundary line corresponds to two separate boundary segments from the two regions.  This will contract with Assumption~\ref{assump:SurfaceConditions}(e). Therefore, at least one of the minimal regions must correspond to an occluded part of the surface.
\end{proof}

\begin{lemma}
\label{lemmar:2}
If a minimal region corresponds to an occluded part of the surface, then it must have an opposite orientation as the visible part that corresponds to the same minimal region.
\end{lemma}
\begin{proof}
According to Assumption~\ref{assump:SurfaceConditions}(d), the minimal region must be incident with a contour line. According to Assumption~\ref{assump:SurfaceConditions}(c), the contour line corresponds to a curve on the surface that connects the visible part and the occluded part associated with the minimal region. Therefore, the visible part and the occluded part must have opposite orientations.
\end{proof}

We can now prove Proposition~\ref{prop:Orientation}:
\begin{proof}
    Suppose two minimal regions $S_1$ and $S_2$ share a boundary line. According to Lemma~\ref{lemmar:1}, one of them must correspond to an occluded part of the surface. Without loss of generality, we assume that $S_1$ corresponds to an occluded part $S_{1}^{\text{o}}$. We further denote the visible surface parts for $S_1$ and $S_2$ as $S_{1}^{\text{v}}$ and $S_{2}^{\text{v}}$, respectively. 
    
    We first note that $S_{2}^{\text{v}}$ must be connected with either $S_{1}^{\text{o}}$ or $S_{1}^{\text{v}}$. Otherwise, due to Assumption~\ref{assump:SurfaceConditions}(c), the shared boundary line must also correspond to a boundary segment of $S_{2}^{v}$, which will contract with Assumption~\ref{assump:SurfaceConditions}(e).
    
    If $S_1$ and $S_2$ share a visible boundary line, then $S_{2}^{\text{v}}$ must be connected with $S_{1}^{\text{o}}$. Therefore, $S_{1}^{\text{o}}$ and $S_{2}^{\text{v}}$ have the same orientation. From Lemma~\ref{lemmar:2}, $S_{1}^{\text{v}}$ and $S_{2}^{\text{v}}$ must have opposite orientations, which proves Proposition~\ref{prop:Orientation}(a).
    
    Similarly, if $S_1$ and $S_2$ share an occluded boundary line, then $S_{1}^{\text{v}}$ and $S_{2}^{\text{v}}$ must be connected, and they have the same orientation. This proves Proposition~\ref{prop:Orientation}(b).
\end{proof}

\section{User Study Questions and Answers} 
\label{appx:answers}

Below are the answers that the participants can choose for each question in the user study. The number before each answer is its score.

\begin{enumerate}[label=Q\arabic*:]
 \item Is the system intuitive compared to other 3D modeling tools you have used, such as Blender, Sketch-up, and Rhinoceros?

\begin{itemize}
    \item 1: Not at all.

    \item 2: Below average level.

    \item 3: Average level.

    \item 4: Slightly above average level.

    \item 5: Well above average level.
 \end{itemize}

 \item Does the generated shape meet your expectations?
\begin{itemize}
 \item 1: Not at all.

 \item 2: The shape is overall close to the expected shape but lacks many details.

 \item 3: The shape is overall close to the expected shape but lacks some details.

 \item 4:
 The shape meets most of my expectations.

 \item 5: The shape meets all my expectations.
\end{itemize}

 \item Are the stroke lines effective in helping you create a model?

\begin{itemize}
 \item 1: Not at all.

 \item 2: They help to control local panel layouts but also affect other areas severely.

 \item 3: They help to control local panel layouts well but can sometimes affect other areas.

 \item 4: They are effective in helping me create a model when used in conjunction with depth samples.

 \item 5: They are very effective in helping me create a model.
 \end{itemize}

 \item Does the feature of drawing occluded regions help your modeling?

\begin{itemize}
 \item 1: Not at all.

 \item 2: It can generate occluded regions but with unreasonable shapes.

 \item 3: It can generate occluded regions but the shape is not always reasonable.

 \item 4: It can generate occluded regions that basically meet my expectations.

 \item 5: It can generate occluded regions that always meet my expectations very well.
 \end{itemize}
 
 \item Do the feature lines help you control the PQ mesh layout?

\begin{itemize}
 \item 1: Not at all

 \item 2: They help to control the layout in some minor areas.

 \item 3: They help to control the layout roughly.

 \item 4: They help to control the layout and basically meet my requirements.

 \item 5: They help to control the layout and are more effective than existing modeling software.
\end{itemize} 

\end{enumerate}

Table~\ref{tab:Result of the user study} shows answers from all the participants and the average score for each question. 

\begin{table}[!hbt]
\caption{Answers from all participants for the user study questions.}
\centering
\label{tab:Result of the user study}
\begin{tabular}{c c c c c c}
\hline
Score & Q1 & Q2 & Q3 & Q4 & Q5\\
\hline
I & 4 & 5 & 4 & 5 & 4 \\
II & 3 & 3 & 4 &5 &4 \\
III& 5 & 5 & 5 & 5 & 5 \\
IV& 5 & 3& 4 &3 &4 \\
V & 4 & 4 &4 &4 &4 \\
VI& 3 & 4 & 3 & 3 & 4\\
VII& 4 &4 &4 &4 & 5 \\
VIII & 4 & 4 &4 & 4 & 5\\
IX & 4 & 4 & 4 &5 & 4 \\
X & 4 & 4 & 5 &4 & 4 \\
XI & 3 & 4 & 4 &4 & 4 \\
XII & 3 & 3 & 3 & 4 & 4 \\
\hline
Mean & 3.84 & 3.92 & 4 & 4.07 & 4.23 \\
\hline
\end{tabular}
\end{table}

\begin{figure*}[t]
 \centering
 \includegraphics[width=\linewidth]{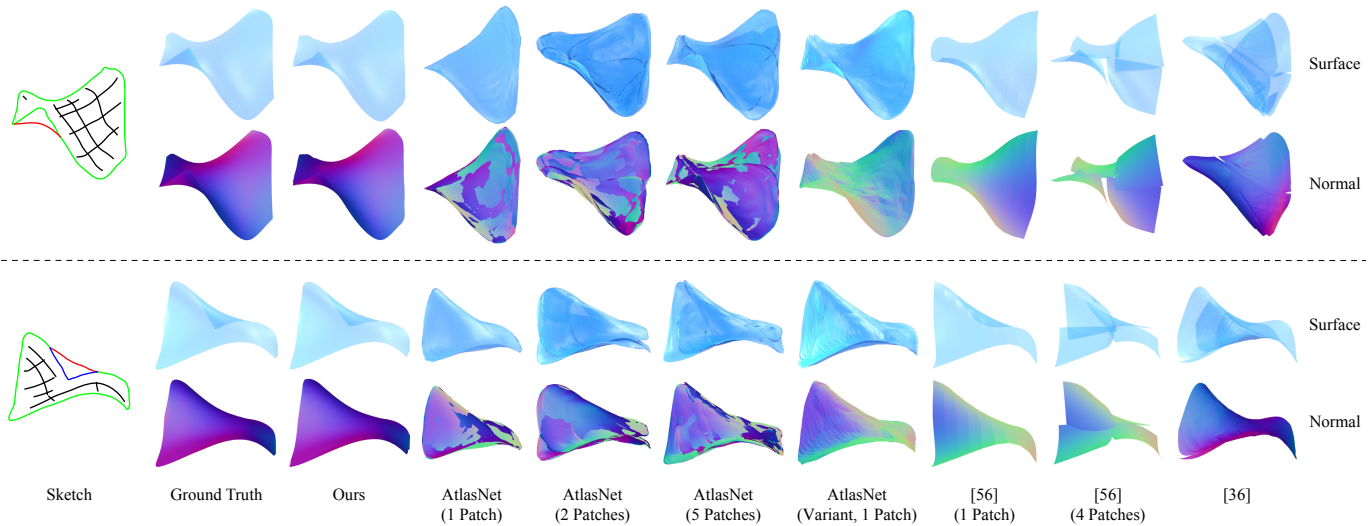}
 \caption{Two sketches from our test dataset, and the surface prediction results using our method, AtlasNet~\cite{atlasnet2018} (with one, two and five patches, and a variant with one patch), the method from~\cite{Deng:2020:3DV} (with one and four patches), and the method from~\cite{smirnov2021patches} (with 10$\times$10 Coons patches).}
 \label{fig:SurfPredictionComparison}
\end{figure*}

\section{Comparison of Surface Prediction Results}
\label{appx:SurfPredictionComparison}
Fig.~\ref{fig:SurfPredictionComparison} shows two sketches from our test dataset, and the predicted surfaces using our method, AtlasNet~\cite{atlasnet2018}, and the methods from~\cite{Deng:2020:3DV} and \cite{smirnov2021patches}.
We also implemented a variant of AtlasNet with a modified loss function: we sample the boundary and the interior of the generated surface and the ground-truth surface respectively, and replace the Chamfer distance term in the original loss function with a combination of two Chamfer distance terms computed using the boundary sample points and the interior sample points respectively; this variant provides better coverage of sample points on the surface boundary, which can improve the accuracy of the boundary shape.
For each result, we also visualize the normals by rendering them as RGB colors over the surface. We can see that our method recovers a surface shape that is close to the ground truth, while the results produced from the other approaches deviate notably from the ground truth. Some results from the other approaches also exhibit incorrect fold-overs. This is potentially because their loss function only uses the Chamfer distance to align the generated shape with the ground truth: a fold-over shape with each ``layer'' close to the ground-truth shape may still achieve a low Chamfer distance.

\begin{figure*}[t]
 \centering
 \includegraphics[width=\linewidth]{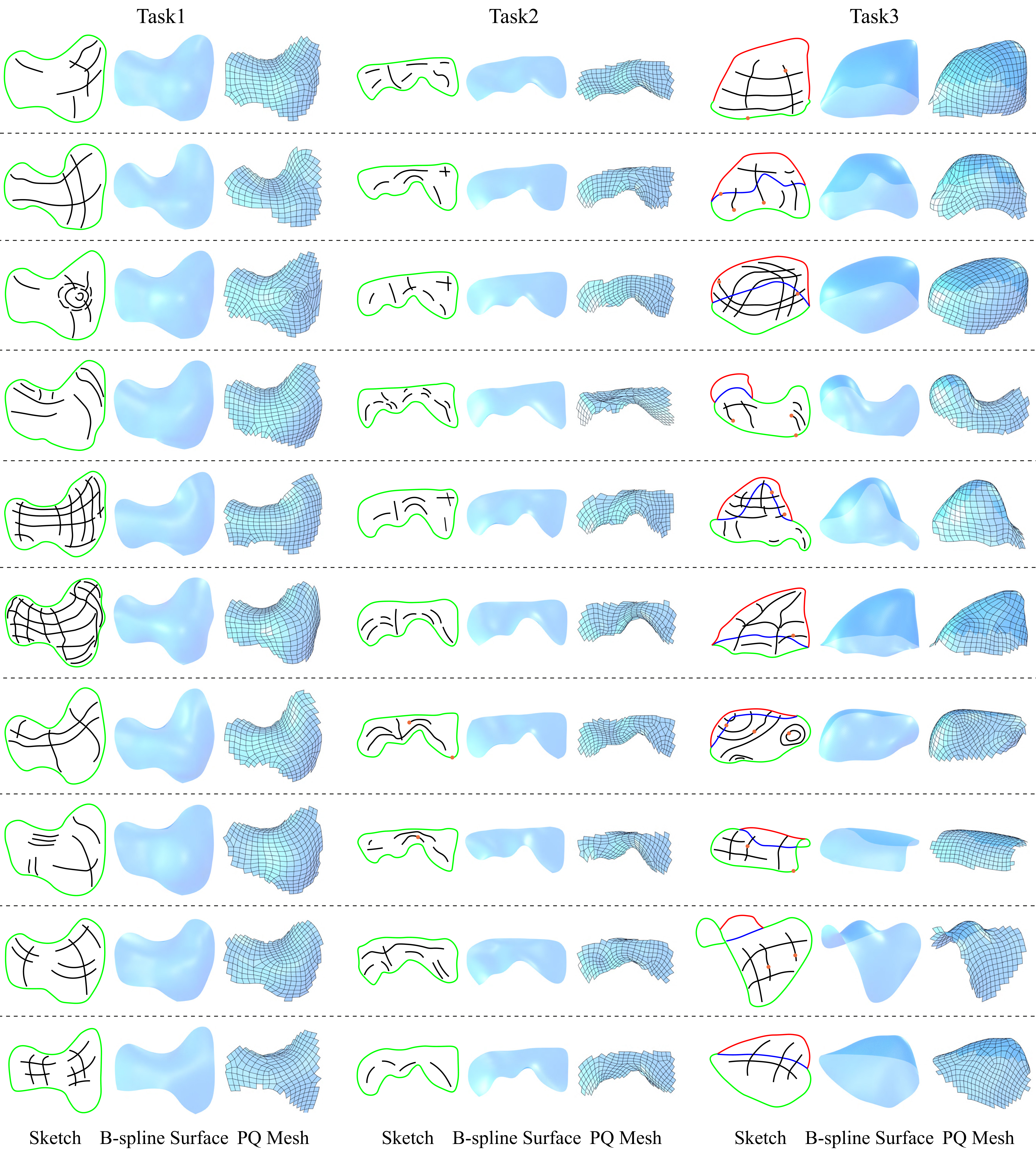}
 \caption{Results of the surface modeling tasks besides the ones already shown in Fig.~\ref{fig:UserStudyExamples}. Each row shows the results from one participant}
 \label{fig:UserStudyRemainingExamples}
\end{figure*}

\section{More Results from Surface Modeling Tasks}
\label{appx:UserStudyMoreExamples}
Fig.~\ref{fig:UserStudyRemainingExamples} shows the participants' designs for the three surface modeling tasks in the user studies, in addition to the ones shown in Fig.~\ref{fig:UserStudyExamples} in Section~\ref{subsec:userstudy}.

\begin{figure*}[t]
\centering
\includegraphics[width=\linewidth]{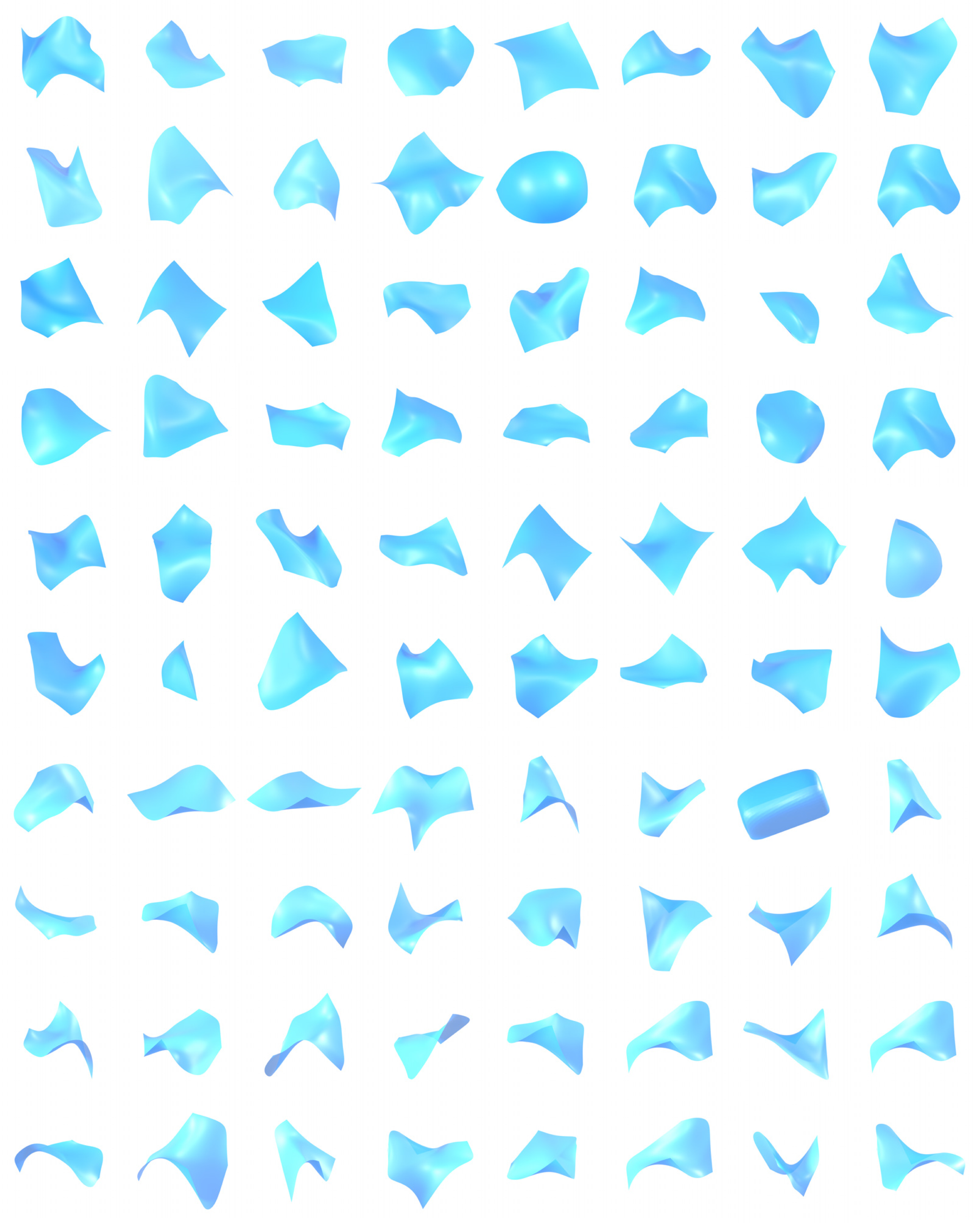}
 \caption{More examples of B-spline surfaces from our dataset.}
\label{fig:More_3d_examples}
\end{figure*}

\section{More Shapes from Our Dataset}
\label{appx:More3dData}
Fig.~\ref{fig:More_3d_examples} shows more examples of B-spline surfaces from our dataset.

\section{Network Training Details for Comparison with~\cite{atlasnet2018,Deng:2020:3DV,smirnov2021patches}} 
\label{exp:training_details}
The networks of \cite{atlasnet2018}, \cite{Deng:2020:3DV} and  \cite{smirnov2021patches} can learn the generation of surface shapes from images. We adopt this setting for all three methods, using their open-source implementation\footnote{\url{https://github.com/ThibaultGROUEIX/AtlasNet}}$^{,}$\footnote{\url{https://github.com/GentleDell/Better-Patch-Stitching}}$^{,}$\footnote{\url{https://github.com/dmsm/LearningPatches}} for training and testing. All parameters are set to their default values given in the open-source implementations.
For all three methods, the training requires sample point positions on the ground-truth surface. Some methods also require sample point normals.
For this purpose, we regularly sample $100 \times 100$ parameter values over the parameter domain of the ground-truth B-spline surface, and use their corresponding point positions and normals on the surface for the training.
Both \cite{atlasnet2018} and  \cite{Deng:2020:3DV} represented the generated shape as a collection of patches. 
We train three networks for~\cite{atlasnet2018} using one, two, and five patches, respectively.
We also train two networks for~\cite{Deng:2020:3DV} using one and four patches, respectively.
In~\cite{smirnov2021patches}, the generated shape is represented using a template consisting of multiple Coons patches, and the default template has a sphere topology. Since we want to generate B-spline surfaces of disk topology, we follow the description in \cite{smirnov2021patches} to construct a disk-topology template consisting of  $10 \times 10$ Coons patches and use it for the training.
\end{document}